\documentclass[11pt]{article}

\usepackage{hyperref}
\usepackage[margin=1in]{geometry}
\usepackage{float}
\usepackage{adjustbox}
\usepackage{booktabs}
\usepackage{xcolor}
\usepackage{amsmath}
\usepackage{amssymb}
\usepackage{caption}
\usepackage{amsthm}
\usepackage{subcaption}
\usepackage{tikz}
\usepackage{lscape}  
\usepackage{setspace}
\usepackage{authblk}

\newtheorem{assumption}{Assumption}[section]
\newtheorem{example}{Example}[section]
\newtheorem{proposition}{Proposition}[section]
\newtheorem{corollary}{Corollary}[section] \newtheorem{definition}{Definition}[section]
\newtheorem{remark}{Remark}[section]
\usepackage[
    backend=biber,
    natbib=true,
    style=apa,
    uniquename=false,
    url=false, 
    doi=true,
    eprint=false
    ]{biblatex}

\AtEveryBibitem{\clearfield{issn}}

\title{An ambit field framework for the full panel of day-ahead electricity prices}
\author[1,2]{
Thomas K. Kloster
}
 \affil[1]{Department of Economics and Business Economics, Aarhus University}
 \affil[2]{CoRE, Center for Research in Energy: Economics and Markets}
\begin{document}
\maketitle

\begin{abstract}
This paper considers the often overlooked fact that electricity spot prices in individual European generation zones evolve as a high dimensional panel structure. A general continuous time framework is developed by formulating the panel as an ambit field indexed by a cylinder surface, where the cross sectional dimension is represented by a circle. This requires a treatment of ambit fields on manifolds, but the departure from Euclidean space allows for embedding intrinsic dependence structures into the index set in a flexible and parameter-free way, where the daily delivery periods have a canonical mapping onto the circle. The model is a natural space-time extension of volatility modulated Lévy-driven Volterra processes, which have previously been studied in the context of energy markets, and the pricing of electricity derivatives turns out to be essentially as analytically tractable as in the null-spatial setting. The space-time framework extends the scope of possible derivatives to products written on individual delivery periods, where spreads between these constitute an interesting example. We establish useful formulas for the pricing of various derivatives along with a simulation scheme, and study the model dependence structure in detail.
\end{abstract}

\clearpage
\section{Introduction}
The liberalised and deregulated European power markets are by now a cornerstone of the electricity infrastructure in Europe, and helps ensure competitive price determination based on demand and supply. The most basic constituent of the market is the day-ahead (or simply spot) market, which settles the prices in each delivery period and generation zone for the coming day. In addition to this spot market, there is also a large and liquid market for various derivatives such as futures, options, and power purchase agreements, which offer important tools for risk management (see \textcites{BenthMonograph}{AidMonograph}). A unique feature of the electricity market is that, although the spot prices are quoted in ordered intervals of the coming day, such as the price of electricity delivered in the hourly intervals 00-01, 01-02 etc., these prices are \emph{determined and revealed} simultaneously and do therefore not have a natural notion of causal direction between them. Since October 2025, the prices in many European zones are on a quarter-hour basis, such that each day has 96 unique prices that are all determined simultaneously. Electricity spot prices therefore evolve as a rather high dimensional panel structure, and not as a conventional time series. In the literature on electricity derivatives, this panel structure is often overlooked in favour of modelling the average spot price, which may be considered a proxy for the overall spot price of electricity and also servers as the underlying of conventional futures contracts and, by extension, derivatives written on these. The purpose of this paper is to extend the traditional univariate \emph{continuous time} modelling framework to the full multivariate panel setting in a consistent and parsimonious way, while keeping enough structure and tractability to estimate the model on data and price and hedge derivatives. A consistent modelling framework for the full panel of prices allows practitioners and market participants to treat prices and products down to the individual delivery periods, and thereby assess their risk exposure and manage their positions more accurately. This is important in the modern market for electricity, where the generation is increasingly driven by volatile renewable sources and where the prices in single delivery periods can easily deviate substantially from the daily average. In the econometric strand of the literature, the panel structure has been studied, but the focus has primarily been on the temporal properties of the price series. For example, \textcite{Huisman} find that individual delivery periods have substantially different characteristics such as differing mean-reversion rates, and \textcite{ErgemenHaldrup2016} find evidence of long-range dependence in the time series. Further motivation for studying the full panel is given in \textcite{Raviv}, who argue that the prices in individual delivery periods contain useful predictive information, even for the average daily price.

It is of course possible to extend existing univariate models for derivatives pricing to a multivariate setting, by simply modelling multiple univariate processes. This is the approach taken in \textcite{VeraartVeraart2014}, who study a multivariate version of the model proposed in \textcite{BARNDORFF-NIELSENOLEE.2013Mesp}, although electricity derivatives are not explicitly considered. This approach is, however, not well-suited for modelling arbitrary or changing partitions of delivery periods, which is relevant as these are not fixed throughout time or generation zone. Furthermore, this approach requires careful specification and restriction of the resulting models to facilitate estimation, and quickly introduces an abundance of parameters. On the other hand, over-parametrized machine learning models are frequently and successfully utilized for forecasting in these settings (see e.g. \textcites{Forecasting1}{Forecasting2} for an overview and comparison of such methods), but suffer from the usual problems of lack of interpretability and few risk management tools for market participants who operate on the market for electricity derivatives, as it is not clear how to price or hedge derivatives in such models. 

Before introducing the proposed modelling framework, we briefly review the three main approaches to the pricing and hedging of electricity derivatives in the literature. The first and direct approach is to specify a reduced form model for the average spot price process as in \textcites{LuciaSchwartz2002}{GemanRoncoroni2006}{BenthKallsen2007}{Meyer-BrandisTankov2008}{BenthBNS}{BARNDORFF-NIELSENOLEE.2013Mesp}{BenthCARMA} and then derive prices of derivatives from this. The second approach of \textcites{BenthKoekebakker2008}{Barndorff-NielsenOleE.2014MEFb}{BenthKruhner1}{BenthParaschiv2018}{BenthVargiolu}{CuchieroEnergy2024} works around this by modelling directly the futures price under a pricing measure. It therefore removes a modelling layer which simplifies computations for derivatives written on the futures, but decouples the futures modelling from the spot price formation process. The third approach is the so-called structural models of \textcite{CarmonaStructural2}, where the average price process is typically specified as a function of some underlying factors, such as the prices of generation variables. Detailed derivations and comparisons of the various approaches may be found in \textcites{BenthMonograph}{CarmonaStructural}{AidMonograph}.
Our approach can be seen as a generalization of the first approach, where a model is specified for the full panel of spot prices, and prices of derivatives derived from this, which allows for characterizing the market price of risk via an appropriate change of measure. In fact, we find that the futures pricing formulas in our tempo-spatial setting are natural extensions of the formulas obtained in the univariate model of \textcite{BARNDORFF-NIELSENOLEE.2013Mesp}, with the added advantage that we may also derive prices of products written on individual delivery periods.

To summarize the preceding discussion, one of the main challenges in modelling the full panel of spot prices is to specify a consistent dynamic model with relatively few parameters that accommodates sufficiently flexible dependence -- both in time and cross-section -- while still being tractable enough to allow for the pricing and hedging of derivatives. To solve this challenge, the present paper introduces a continuous time model for the panel of spot prices that takes the natural cyclical nature of the market into account, by considering the spot price panel as a realization of a random field indexed by a continuum of days and delivery periods. This approach allows us to treat the full panel of prices using quite few parameters, while trivially accounting for arbitrary partitions of delivery periods. Concretely, we represent each delivery period as a point $h$ on a circle, and for a given time $t$ and delivery period $h$, we model the spot price $S_{t}(h)$ as an ambit field
\begin{equation}\label{eq:model}
S_{t}(h) = \int_{A_{t}(h)}\kappa (t,s,h,\xi)a_{s}(\xi) c(ds,d\xi) + \int_{A_{t}(h)}K(t,s,h,\xi)\sigma_{s}(\xi) L(ds,d\xi).
\end{equation}
Here, $L$ is a Lévy basis representing the price shocks in both the temporal and cross sectional dimension, $c$ is a measure to be specified, $\kappa,K$ are kernel functions that govern the dependence structure, and $a,\sigma$ are random fields that permit us to model stochastic seasonality and volatility effects. The crucial modelling ingredient that gives rise to a random field are the ambit sets $A_{t}(h)$, which, for each $(t,h)$, are truncated cylinder surfaces in $\mathbb{R}^{3}$ parametrized as
\begin{equation}\label{eq:ambit_set}
A_{t}(h) = \lbrace (s,\cos (\theta),\sin(\theta)) \mid s\in (-\infty,t], \; \theta \in (0,2\pi] \rbrace.
\end{equation}
This form of the ambit set imposes an underlying, but parameter-free, dependence structure on the delivery periods, as it places them on a circle in the plane, giving rise to an intrinsic and intuitive notion of distance between delivery periods. In particular, we explicitly impose that the last delivery period of a given day, $t$, is ``closer'' to the first delivery period of the following day $t+1$, than it is to the first delivery period on day $t$, which is a quite natural condition when considering that prices are based on electricity supply and demand in \emph{actual time}.

The framework outlined in \eqref{eq:model} is very general and embeds most of the previously proposed reduced form models if one considers the null-spatial case, which corresponds to the ambit type processes considered in \textcites{BARNDORFF-NIELSENOLEE.2013Mesp}{Barndorff-NielsenOleE.2014MEFb}{VeraartVeraart2014}{Bennedsen2017}. In particular, the framework nests any one-factor Lévy-driven Ornstein-Uhlenbeck process via the null-spatial kernel $K(t,s)=\mathbf{1}_{[0,t]}e^{-\lambda (t-s)}$, as well as the CARMA type processes of \textcites{KluppelbergCARMA}{BenthCARMA}. Besides generalizing many previously studied models to the panel data setting, the main novelty in terms of modelling lies in specifying the ambit sets in a way that naturally permits us to treat the underlying market structures, while maintaining the same degree of analytical tractability as in well-known univariate models. This also represents one of the first applications of ambit fields within finance to the genuine tempo-spatial setting, and in particular to more general spaces than simply Euclidean spaces.

We consider the first term in \eqref{eq:model} to represent non-stationarities such as seasonality and trend effects, where we note that it is possible to model these as also containing a stochastic component. Upon suitable de-seasonalization, we are then left with the second term, which should contain the remaining stylized facts of spot prices, such as extreme price spikes, potential long memory, mean-reverting behaviour, and occasional negative prices. The presence of negative prices becomes much more relevant when modelling individual price series compared to the average price process, and we therefore model the price level directly as opposed to the log-price. By choosing the model ingredients $L,K,\sigma$ appropriately, we may include these stylized facts while obtaining a suitably stationary field with a known moment structure. Correlation and memory properties in prices are governed by the kernel function $K$, and we consider how to choose this crucial quantity. To avoid model misspecification, we show that it is possible to specify $K$ in a semi-parametric way, such that we may learn the correlation structure from the data. The semi-parametric kernel is based on a generalization of the temporal kernel studied in \textcite{Bennedsen2017}, which in principle allow us to assess, via estimation, whether the panel is consistent with a semimartingale property for the ambit field in time. We set up a toy study based on German data where the semi-parametric kernel is estimated via a Whittle likelihood approach. The findings indicate that the semimartingale property is violated, which supports previous findings in the literature. As we have a view towards the world of derivatives, we also derive expressions for the prices of futures contracts and study how to ensure some stylized properties of these. This is accompanied by a structure preserving change of measure in the model, which serves as a first step towards characterizing the dynamics of the market price of risk in each delivery period. Given that we have a model for each delivery period, we also turn to the problem of pricing derivatives written on spreads between subsets of delivery periods. Such contracts  allow market participants to hedge risks associated to individual delivery periods, which is an aspect not covered in the conventional null-spatial setting. An appealing property of our framework is that, by treating the full panel of prices as one modelling object, the pricing of such \emph{within-day} products is automatically self-consistent and arbitrage free. Finally, we also provide a simulation scheme for the model, which is inspired by the method introduced in \textcite{Eyjolfsson2015}, but adapted to allow for singular kernel functions. The simulation scheme is used to provide some concrete illustrations of the modelling framework and the pricing of various derivatives. 

The rest of the paper is organized as follows. Section \ref{sec:model_setup} motivates the space-time interpretation of the panel of prices, and Section \ref{sec:ambit_on_manifolds} rigorously introduces the model \eqref{eq:model}, which requires a brief treatment of ambit fields on manifolds. Section \ref{sec:derivatives} treats the problem of pricing and hedging derivatives, deriving a structure preserving change of measure and analyzing both futures, within-day spreads, and leverage. Section \ref{sec:kernel} introduces the semi-parametric model and contains an estimation exercise on German data, where the kernel function is learned from data. Section \ref{sec:simulation} develops a simulation scheme for the proposed model, and several illustrative examples are provided. Section \ref{sec:conclusion} concludes. Mathematical proofs of the propositions throughout the paper are relegated to Appendix~\ref{app:proofs}.

\section{A random field interpretation for the panel of prices}\label{sec:model_setup}
Electricity spot prices in many European generation zones are determined each day on an auction, based on supply and demand. The spot prices are typically divided into hourly, half-hourly, or quarter-hourly prices, with each such timeslot simply being referred to as a delivery period, as the quoted spot price refers to the price of delivery (and thus consumption) of one MWh of electricity during that time. We let $H$ denote the total number of delivery periods in the given market under consideration and enumerate the ordered and non-overlapping delivery periods by $d=1,\ldots ,H$, such that period $d=1$ corresponds to the first delivery period of the day. For a given zone, the prices in all delivery periods are determined and revealed simultaneously on the day before they are valid. If we denote by $P_{t}(d)$ the spot price of one MWh of electricity delivered during period $d$ that is \emph{revealed} at time $t$, we then observe the full cross section $(P_{t}(d))_{d=1}^{H}\in \mathbb{R}^H$ simultaneously at a daily frequency. In this work, we introduce a natural notion of distance between delivery periods and use this to formulate a continuous time model for the full panel of spot prices in which the cross section is treated as one object. This is in contrast to a conventional multivariate approach which inevitably leads to a large amount of model parameters and restrictions on these (consider the simplest case of $H$ independent AR(1) processes, in which case the model already has $2H$ parameters). Our approach has several advantages compared to modelling the average price or using conventional panel data models, of which we highlight the following.
\begin{itemize}
\item The approach utilizes a priori knowledge of the market structure to embed dependence in a model-free manner.
\item Modelling each delivery period in continuous time allows for considering derivatives written on a subset of delivery periods, or even spreads between them, in a consistent manner.
\item The model trivially allows for an arbitrary, or even changing, partition of the day into delivery periods, as well as arbitrary maturities and settlement periods for derivatives.
\end{itemize}

It is well documented that the observed processes $t\mapsto P_{t}(d)$ exhibit strong autocorrelation for fixed $d$, and several approaches to capture this path dependence has been proposed by e.g. \textcites{HaldrupNielsen2006}{LongRange2}{ErgemenHaldrup2016}{Bennedsen2017}. However, the dependence structure in the cross sectional dimension $d\mapsto P_{t}(d)$ has mostly been neglected. 
\begin{figure}
	\centering
	\begin{subfigure}[b]{0.49\textwidth}
		\includegraphics[scale=0.4]{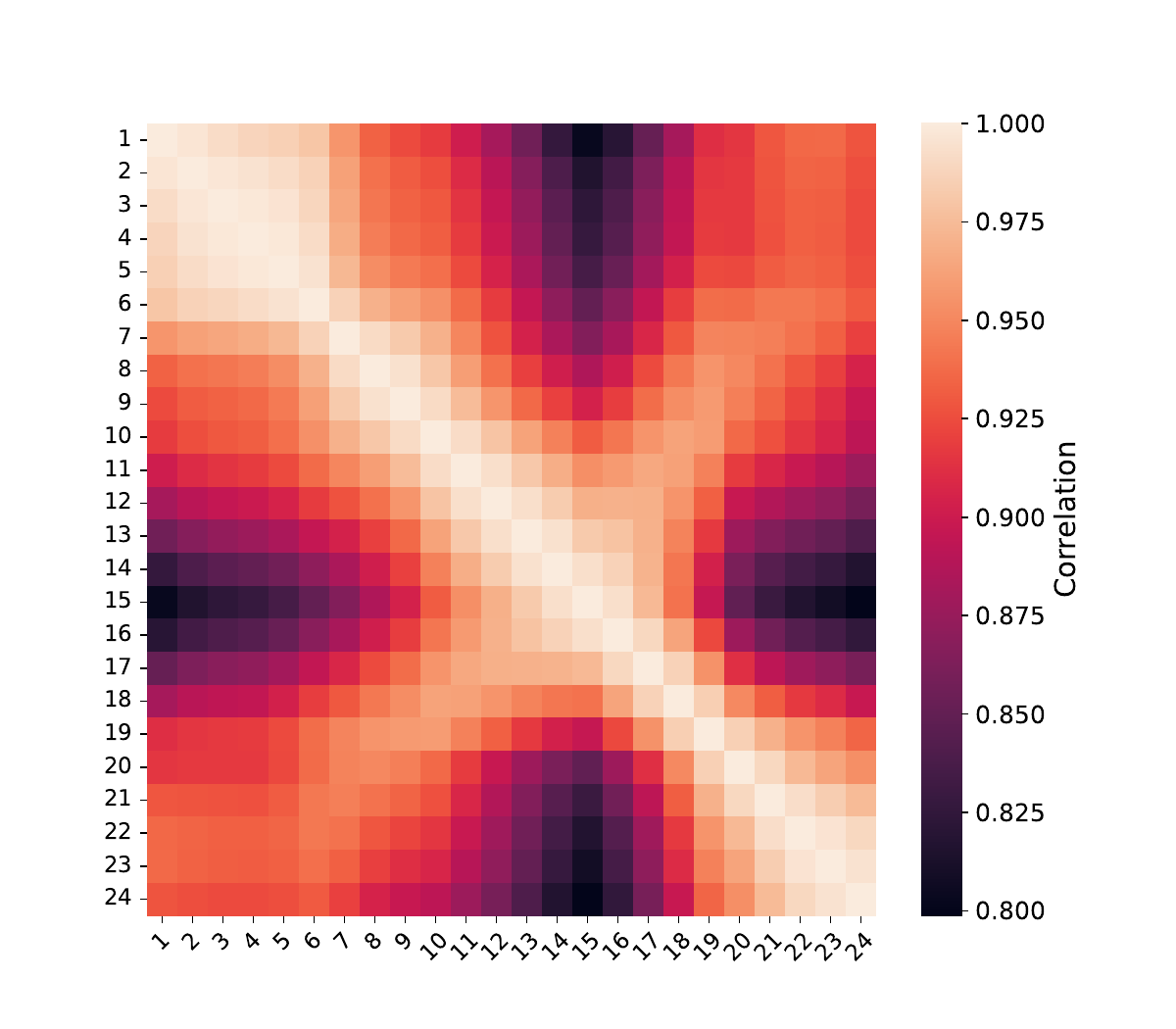}
	\end{subfigure}
	\begin{subfigure}[b]{0.49\textwidth}
		\includegraphics[scale=0.4]{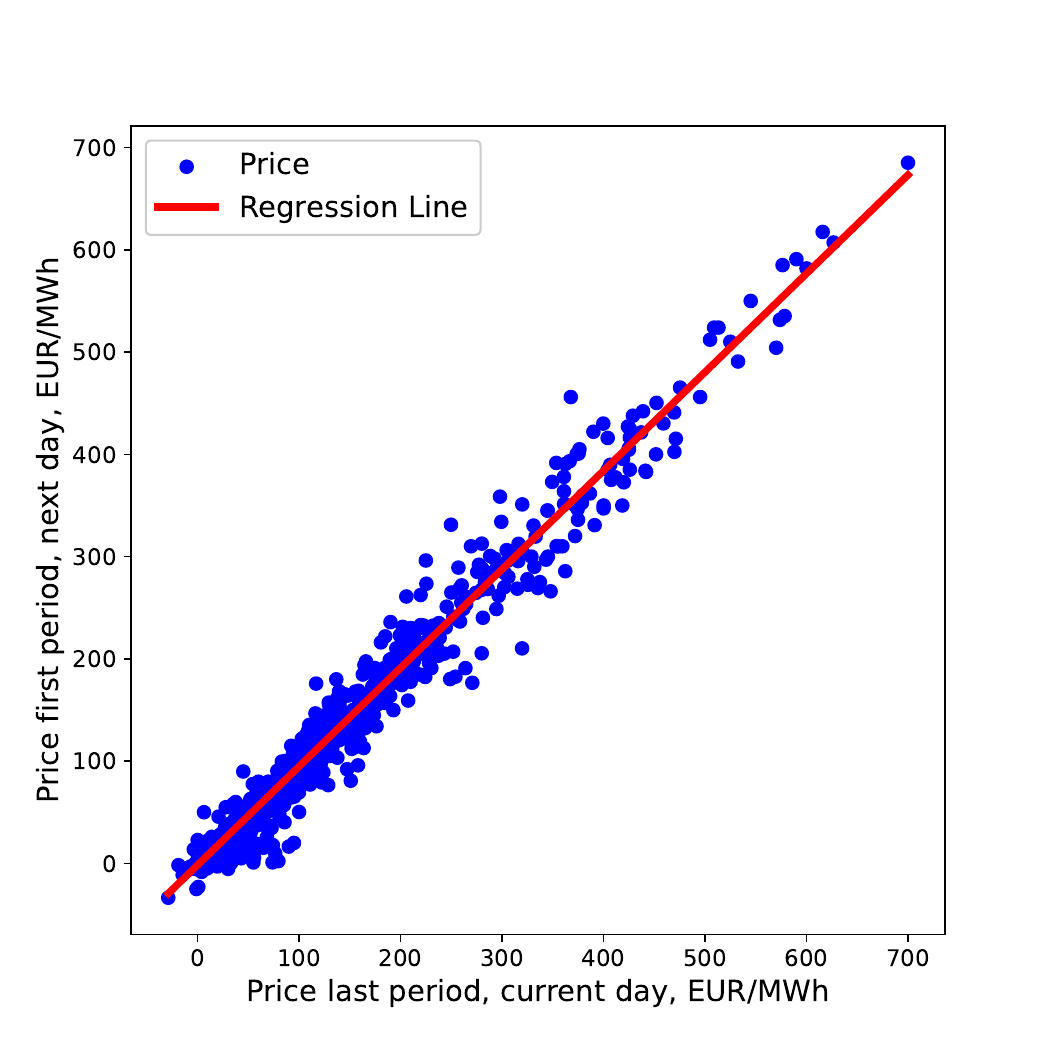}
	\end{subfigure}
	\caption{Correlation matrix of hourly spot prices on the German market from October 1st 2018 to October 1st 2024 (left), and regression of spot price during the first delivery period on the last spot price of the preceding day (right). Prices are shown in Euros (EUR) per megawatt hour (MWh).}\label{fig:correlation_matrix}
\end{figure}
As spot prices are determined based on supply and demand, it stands to reason that prices are highly correlated across hours that are close to each other \emph{in actual time}, as supply and demand rise and fall rather continuously across actual time. Figure \ref{fig:correlation_matrix} illustrates that, empirically, the full daily cross section is very highly and positively correlated, but also that periods that are close ``on the clock'' are more highly correlated. Figure \ref{fig:correlation_matrix} also shows that the last delivery period each day has very strong predictive power on the price in the first delivery period on the following day, as these two periods are close in actual time, in spite of being determined and revealed on different days. This entails a type of transversality condition that connects the spot price in the last delivery period, $P_{t}(H)$, to the spot price of the first delivery period of the following day, $P_{t+1}(1)$. We refer to this connection between late delivery periods on the preceding day and early delivery periods on the following day as \emph{cyclicality}. We summarize the stylized dependence structure in the panel in the following three conditions.
\begin{enumerate}
\item \textbf{Adjacency:} $P_{t}(d)$ and $P_{t}({d\pm 1})$ are highly correlated  for $d\in \lbrace 2,\ldots ,H-1\rbrace$.
\item \textbf{Cyclicality:} $P_{t}(1)$ and $P_{t-1}(H)$ are highly correlated. 
\item \textbf{Autocorrelation:} $P_{t-1}({d})$ and $P_{t}({d})$ are highly correlated.
\end{enumerate}
It may be possible to enforce a correlation structure that adheres to these three conditions by imposing a suitable multivariate model as in e.g. \textcite{VeraartVeraart2014}, but in this paper we embed the dependence structure in the model itself. A way of doing so, is to map each daily delivery period onto a circle via $d \mapsto\left( \cos \left(\frac{2\pi d}{H}\right), \sin \left(\frac{2\pi d}{H} \right) \right)$. All observation time points $(t,d)$ then correspond naturally to a unique point on a cylinder surface $\mathcal{C}$ in $\mathbb{R}^{3}$ as follows.
\begin{equation}\label{eq:C_set}
(t,d) \mapsto \mathcal{C} = \left\lbrace (t, \cos (\theta), \sin(\theta))\mid t\in \mathbb{R}, \; \theta\in (0,2\pi] \right\rbrace.
\end{equation}
Letting $h$ denote a point on the unit circle, we may then represent the panel of spot prices as a random field indexed by $\mathcal{C}$, $(S_{t}(h))_{(t,h)\in\mathcal{C}}$, where we note that $h=(h_1,h_2)$ is two-dimensional. We have changed the notation for spot prices to $S_{t}(h)$ to emphasize that this is an interpretation of the observed panel $P_{t}(d)$. In this view, we then observe a ``slice'' (or rather, equispaced discrete points on a slice) of electricity prices at a daily frequency, where the highly correlated delivery periods are close in terms of distance on $\mathcal{C}$. Note that we have simply normalized the radius of the cylinder surface $\mathcal{C}$ to $1$, as this simplifies the exposition. The index set $\mathcal{C}$ and the observations along this are illustrated in Figure \ref{fig:index_set}. By considering the panel as point-samples of a $\mathcal{C}$-indexed random field, we pass from a multivariate setting to an infinite dimensional setting, where both time and space is a continuum. The discrete observations on $\mathcal{C}$ then correspond to sampling the latent spot price process at fixed points in space-time. In the following we shall see that this simplifies the analysis since we only require the treatment of one infinite dimensional object, as opposed to many one-dimensional objects.
\begin{figure}[ht]
    \centering
    \begin{subfigure}[b]{0.45\textwidth}
        \centering
        \includegraphics[width=\textwidth]{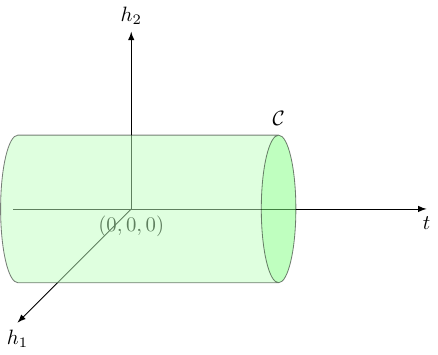}  
    \end{subfigure}
    \hspace{0.05\textwidth}
    \begin{subfigure}[b]{0.45\textwidth}
        \centering
        \includegraphics[width=\textwidth]{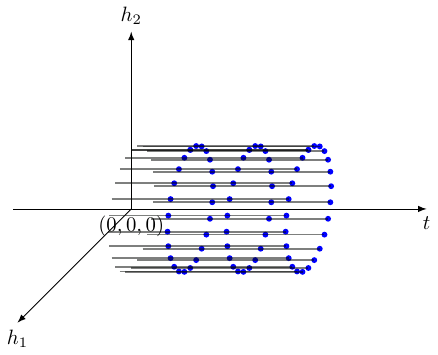}  
    \end{subfigure}
    \caption{Visualizations of the index set $\mathcal{C}$ (left), and of $H=24$ observations along slices of $\mathcal{C}$ represented by blue dots (right).}
    \label{fig:index_set}
\end{figure}
\begin{remark}
One could envision a stochastic process indexed by a helix, which wraps around the cylinder surface $\mathcal{C}$ to describe the same stylized properties. Such a structure would naturally also exhibit adjacency and cyclicality, but it also enforces a causal direction  between the same-day delivery periods, as information can then only propagate forwards in time. A helix would therefore not lead to a good description of the panel of prices, as it is crucial that the prices across all same-day delivery periods are based on the same information set. 
\end{remark}

\section{Ambit fields on manifolds}\label{sec:ambit_on_manifolds}
The interpretation of the panel of spot prices evolving as a random field indexed by the surface $\mathcal{C}$ provides a clear avenue for defining a model for the panel, namely by defining a suitable random field that may reproduce the desired stylized empirical facts. We will utilize the powerful theory of ambit stochastics to define $(S_{t}(h))_{(t,h)}$ as an ambit field on $\mathcal{C}$, which is a natural generalization of Lévy driven Ornstein-Uhlenbeck (OU) processes which have been popular as models for the average spot price (see e.g. \textcite{BenthMonograph}). We refer to the monograph \textcite{Ambit} and the references therein for a comprehensive introduction to ambit stochastics. Working on the cylinder surface $\mathcal{C}$ entails working with an ambit field on a 2-dimensional manifold embedded in $\mathbb{R}^3$, which is slightly different than most of the literature on ambit fields that primarily treats the Euclidean setting. This section is therefore devoted to carefully defining ambit fields on manifolds and the related properties that we will need. To this end, we require the notion of a Lévy basis, which may be considered as a generalization of Lévy processes to random measures, and these will act as the noise in the model. For more details on Lévy bases, we refer to \textcites{Rosinski1989}{Urbanik1966}{Pedersen2003}. In the following, we fix a probability space $(\Omega ,\mathcal{F},\mathbb{P})$ on which all random variables are defined, and let $(M,\mathcal{M})$ be a smooth Riemannian $k$-manifold where $\mathcal{M} = \mathcal{B}(M)$ is the Borel $\sigma$-algebra on $M$. Further, we define $\mathcal{M}_{b}$ as the $\delta$-ring generated by all totally bounded sets in $M$, on which Lévy bases will be defined. We shall always assume that $M$ is connected, oriented, and without boundary. A Riemannian manifold is equipped with a metric, $\rho$, which is a bilinear form on the tangent space $T_{p}M$, such that the map
$
M\ni p \to \rho_{p}(X_p,Y_p)\in \mathbb{R}
$,
is smooth for all locally defined smooth vector fields $X$ and $Y$. In our applied context, we shall adopt the extrinsic viewpoint that $M$ is embedded in $\mathbb{R}^n$ with smooth, injective embedding $\iota$, and that the metric $\rho$ is induced by this embedding. The metric is thus represented by the matrix $(\rho_{ij})_{1\leq i,j\leq k}$ with entries 
\[
\rho_{ij}(p) = \left\langle \frac{\partial\iota_{p} }{\partial u^{i}},\frac{\partial\iota_p }{\partial u^{j}} \right\rangle_{n}, \quad 1\leq i,j\leq k,
\]
where $u^{1},\ldots, u^{k}$ are local coordinates on $M$ and $\langle \cdot, \cdot\rangle_n$ is the usual inner product on $\mathbb{R}^n$. From now on, we shall write $\rho_M$ as the metric associated to $M$, where it is implicitly understood that $\rho_M$ is induced by the embedding of $M$ in $\mathbb{R}^n$, which is in turn determined by the application at hand. The Riemannian metric $\rho_M$ then defines the Riemannian measure $\lambda_M$, which allows us to do integration on the manifold. The Riemannian measure $\lambda_M$ on $M$ is such that
\begin{equation}\label{eq:volume_integral}
\int_{A}f(z) \lambda_{M}(dz) = \int_{U_A} f(\iota (u)) (\det \rho_{M}(u))^{\frac{1}{2}}\lambda^{k}(du), \quad A\in \mathcal{M},
\end{equation}
where $\lambda^{k}$ denotes the Lebesgue measure on $\mathbb{R}^{k}$ and $U_A=\iota^{-1}(A)$. In the case of the cylinder surface $\mathcal{C}$ given by \eqref{eq:C_set}, the integral in \eqref{eq:volume_integral} is simply a surface integral with a particularly convenient parametrization, and the Riemannian metric $\rho_{\mathcal{C}}$ on $\mathcal{C}$ is just represented by the $2$-dimensional identity matrix with determinant 1. The corresponding Riemannian measure $\lambda_{\mathcal{C}}$ is thus very simple, and it holds that $\int_{A}f(z) \lambda_{\mathcal{C}}(dz) = \int_{U_A}f(u,r(v))\lambda^{2}(du,dv)$, where we have used the notation $r(v)=(\cos (v),\sin(v))$ to denote the unit circle parametrization. To keep the notation light, we shall often use $r(v)$ as the unit circle parametrization throughout.

\subsection{Lévy bases on manifolds}
\begin{definition}[Lévy basis]
A Lévy basis $L$ on $(M,\mathcal{M})$ is a collection of random variables ${\lbrace L(A)\; : \; A\in \mathcal{M}_{b}\rbrace}$, such that 
\begin{enumerate}
\item For any sequence $A_1,A_2,\ldots $ of disjoint sets in $\mathcal{M}_{b}$ with $\bigcup_{j=1}^{\infty}A_j \in\mathcal{M}_{b}$ we have $L\left( \bigcup_{j=1}^{\infty}A_j\right) = \sum_{j=1}^{\infty}L(A_j)$, where the equality and convergence of the sum holds almost surely.
\item $L$ is \emph{independently scattered}, i.e. for disjoint elements $A_{1},\ldots ,A_{n}$ in $\mathcal{M}_{b}$ the random variables $L(A_1),\ldots ,L(A_n)$ are independent.
\item The law of the random variable $L(A)$ is infinitely divisible for any $A\in\mathcal{M}_{b}$.
\end{enumerate}
\end{definition}
Lévy bases are uniquely characterized in law by a so-called characteristic quadruplet (CQ) of the form $(\gamma,\Sigma,\nu,c)$, where $c$ is a measure on $(M,\mathcal{M})$ called the \emph{intensity measure}, $\gamma \in\mathbb{R}$, $\Sigma \geq 0$ are constants, and $\nu$ is a Lévy measure such that $\int_{\mathbb{R}\setminus \lbrace 0\rbrace} \min (1,x^{2})\nu (dx) < \infty$ \citep[see Definition 33 in][]{Ambit}. To any Lévy basis, we may therefore associate an infinitely divisible random variable $L'$, called the \emph{Lévy seed} of $L$, which has triplet $(\gamma,\Sigma,\nu)$ and thus cumulant function of Lévy-Khintchine type
\[
C(u;L') = \mathrm{i}u\gamma - \frac{1}{2}u^2 \Sigma + \int_{\mathbb{R}}\left( e^{\mathrm{i}ux}-1-\mathrm{i}ux\mathbf{1}_{[-1,1]}(x) \right) \nu (dx),
\]
where $C(u;X)=\log (\mathbb{E}[ e^{\mathrm{i}uX} ])$ for an arbitrary random variable $X$ and $u\in\mathbb{R}$. We will only consider CQ's where the elements are constant (except in Proposition~\ref{prop:esscher_transform}), and shall implicitly assume so in the following, although many results extend in a natural way to the inhomogeneous setting where $\gamma(\cdot),\Sigma(\cdot),\nu (dx,\cdot)$ are functions on $M$. The intensity measure $c$ is free to be specified, and may be regarded as a pure modelling parameter, which essentially governs the scaling of the infinitely divisible distribution of the seed $L'$ according to \citep[see][equation 5.10]{Ambit}
\[
C(u;L(A)) = C(u;L')c(A), \quad A \in \mathcal{M}_b.
\]
However, the intensity measure is closely related to stationarity properties of the Lévy basis, which have important implications for estimation and inference. When $(M,\mathcal{M})=(\mathbb{R}^{n},\mathcal{B}(\mathbb{R}^{n}))$, the only choice of $c$ that yields a stationary (translation invariant) Lévy basis is, up to scaling, the $n$-dimensional Lebesgue measure $\lambda^{n}$, which means that $c=\lambda^{n}$ is usually taken for granted. When $M$ is a Riemannian manifold embedded in Euclidean space, there are various operations under which the Lévy basis can be stationary, and the Lebesgue measure is not a feasible intensity measure as $M$ is typically a Lebesgue null set. We therefore require a more general definition of stationarity of Lévy bases than usual, which we provide in Definition~\ref{def:stationary}. It is straightforward to check that Definition~\ref{def:stationary} reduces to the usual translation invariance whenever $(M,\mathcal{M})=(\mathbb{R}^{n},\mathcal{B}(\mathbb{R}^{n}))$ and $G$ is the group of translations on $\mathbb{R}^{n}$. In Proposition~\ref{prop:levy_basis_stationary} we then give necessary and sufficient conditions on the intensity measure $c$ for a Lévy basis on $(M,\mathcal{M})$ to be stationary with respect to arbitrary operations. 
\begin{definition}[Stationary Lévy basis]\label{def:stationary}
Let $(M,\mathcal{M})$ be a Riemannian manifold, and let $G$ be a group of operations from $M$ onto $M$ with function composition being the group operation. Let $g\in G$, and let $A_{1},A_{2},\ldots ,A_{n}$ be a finite collection elements in $\mathcal{M}_{b}$ such that $g(A_1),g(A_2),\ldots, g(A_n)$ belong to $\mathcal{M}_{b}$. A Lévy basis $L$ on $(M,\mathcal{M})$ is said to be stationary with respect to $G$ if
\[
\left(
L(A_1), L(A_2), \ldots, L(A_n)
\right) \overset{d}{=}
\left(
L(g(A_1)), L(g(A_2)), \ldots , L(g(A_n))
\right).
\]
\end{definition}
\begin{proposition}\label{prop:levy_basis_stationary}
Let $M$ be a Riemannian manifold and let $L$ be a Lévy basis on $(M,\mathcal{M})$ with characteristic quadruplet $(\gamma ,\Sigma , \nu, c)$. Let also $G$ be a group of operations from $M$ onto $M$ with function composition being the group operation. Then $L$ is stationary with respect to $G$ if and only if $c$ is invariant with respect to $G$, i.e.
$c(A) = c(g(A))$ for $g\in G$ and $A\in \mathcal{M}$.
\end{proposition}
A consequence of Proposition \ref{prop:levy_basis_stationary} is that, when working on more general spaces than simply $\mathbb{R}^{n}$, we potentially have more freedom to choose the intensity measure, but the viable choices of $c$ depend on the operations under consideration. If we restrict attention to operations $g$ that are isometries on $M$, we have a prime candidate for an intensity measure in the Riemannian measure $\lambda_{M}$ on $M$, as isometries preserve the Riemannian metric, thus making $\lambda_{M}$ invariant with respect to all isometries by definition (see equation \eqref{eq:volume_integral}). The following proposition shows that, on the cylinder surface \eqref{eq:C_set}, the Riemannian measure in fact arises from pushing forward a translation invariant Lévy basis on $(\mathbb{R}^2,\mathcal{B}(\mathbb{R}^2))$ onto $\mathcal{C}$.
\begin{proposition}\label{prop:pushforward}
Let $M=\mathcal{C}$ with parametrization \eqref{eq:C_set}, and denote this parametrization by $\varphi$. Let $L$ be a translation invariant Lévy basis on $(\mathbb{R}^{2},\mathcal{B}(\mathbb{R}^{2}))$ with characteristic quadruplet $(\gamma ,\Sigma ,\nu, \lambda^{2})$. Then the pushforward of $L$, onto $M$ under $\varphi$ is (in law) a Lévy basis $L_{M}$ on $(M,\mathcal{M})$ with characteristic quadruplet $(\gamma ,\Sigma ,\nu,\lambda_{M})$.
\end{proposition}
As is the case of Lévy processes, one may pose the question if the ``Lévy basis in law'' $L_M$ of Proposition~\ref{prop:levy_basis_stationary} has a suitably regular modification. From a fully general and intrinsic viewpoint, the answer is not clear, but Proposition \ref{prop:levy_ito} shows that on the cylinder surface, we may always construct the Lévy basis to admit a Lévy-Itô type decomposition. 
\begin{proposition}\label{prop:levy_ito}
Let $M=\mathcal{C}$ with parametrization \eqref{eq:C_set}, and denote this parametrization by $\varphi$. Then for a given CQ $(\gamma,\Sigma,\nu,\lambda_{M})$, there exists a Lévy basis $L_M$ on $(M,\mathcal{M})$ with such a CQ and with the following Lévy-Itô decomposition
\begin{equation}\label{eq:levy_ito1}
L_M(A) = \gamma\lambda_{M}(A) + W_{M}(A) + \int_{\mathbb{R}}y\mathbf{1}_{(-1,1)}(y)(N_{M}-\ell_M)(d y,A) + \int_{\mathbb{R}}y\mathbf{1}_{(-1,1)^\complement}(y)N_{M}(d y,A),
\end{equation}
where $W_M$ is a Gaussian Lévy basis cf. Example \ref{ex:Gaussian_levy_basis} with CQ $(0,\Sigma,0,\lambda_M)$ and $N_M$ is Poisson Lévy basis with CQ $(0,0,\nu,\lambda_M)$ and compensator $\ell_M$. 
\end{proposition}
\begin{example}[Gaussian Lévy basis]\label{ex:Gaussian_levy_basis}
Consider a Lévy basis $L$ on $(M,\mathcal{M})$ with CQ $(0,\Sigma,0,\lambda_{M})$. Then we call $L$ a Gaussian Lévy basis and it holds that $L(A) \sim N(0, \lambda_{M}(A)\Sigma)$ for $A\in\mathcal{M}_{b}$.
\end{example}
\begin{example}[Normal inverse Gaussian Lévy basis]\label{ex:NIG_levy_basis}
Consider a Lévy basis $L$ on $(M,\mathcal{M})$ with CQ $(\gamma ,0,\nu,\lambda_{M})$ where
\[
\gamma = \mu + \frac{2}{\pi}\delta \alpha \int_{0}^{1}\sinh \left( \beta y\right)K_{1}(\alpha y)dy, \quad
\nu (dx) = \frac{\delta\alpha}{\pi \lvert x\rvert}K_{1}(\alpha \lvert x \rvert)e^{b x}dx,
\]
where $K_{v}$ denotes the modified Bessel function of the second kind with index $v$ and the parameters satisfy $ \mu \in\mathbb{R},\; \delta >0, \; 0\leq \beta < \alpha$. Let $L'$ be an infinitely divisible random variable with cumulant function
\[
C(u;L') = \mathrm{i}u\mu + \delta \left( \sqrt{\alpha^2 - \beta^2} - \sqrt{\alpha^2 - (\beta + \mathrm{i}u)^2} \right), 
\]
which corresponds to a normal inverse Gaussian (NIG) distribution. We say that $L'\sim NIG(\alpha,\beta,\mu,\delta)$. Then we call $L$ a NIG Lévy basis and it holds that $L$ has Lévy seed $L'$, and that
\[
L(A) \sim NIG(\alpha,\; \beta ,\; \lambda_{M}(A)\mu, \; \lambda_{M}(A)\delta ), \quad A\in \mathcal{M}_{b}.
\]
\end{example}
\subsection{Ambit fields}
To define ambit fields, we will rely on the $L^2$ stochastic integral of \textcite{Walsh1986}, which has very simple integrability criteria and allows for sufficiently general integrands, along with extensions to integrals over infinite sets. The construction parallels Itô's theory for square integrable martingales, and for Lévy bases on $(\mathbb{R}_{+},\mathcal{B}(\mathbb{R}_+))$ one indeed recovers the classical theory. All we shall require is the mild assumption that the Lévy seed has finite second moment, which we henceforth assume. We shall also impose that the manifold $M$ has a space-time structure, which is in particular satisfied by the cylinder surface $\mathcal{C}$. 
\begin{assumption}\label{ass:space_time_structure}
We assume that the manifold $M$ is of the form $M=\mathbb{R}\times \Pi$ for a Riemannian manifold $\Pi$. The first coordinate is interpreted as time and the second as space. This entails that the standard Riemannian metric has a product structure, and hence the Riemannian measure on $M$ will be of the form $\lambda_{M}=\lambda^{1}\otimes \lambda_{\Pi}$, where $\lambda_{\Pi}$ is a fixed Riemannian measure on the manifold $\Pi$. 
\end{assumption}
Under Assumption~\ref{ass:space_time_structure}, any Lévy basis $L$ is an orthogonal martingale measure (as defined in \textcite{Walsh1986}) on $(M,\mathcal{M})$ with respect to the right-continuous filtration $\mathcal{F}_{t}^{L}$ generated by the increments of $L$ as
\begin{equation}\label{eq:L_filtration}
\mathcal{F}_t^L = \bigcap_{n=1}^{\infty}\sigma \left( L((a,b]\times A) : \, -\infty<a<b\leq (t+\tfrac{1}{n}), \, A\in \mathcal{B}(\Pi) \right).
\end{equation}
In the rest of the paper, we assume that any filtration is augmented with the $\mathbb{P}$-null sets of $\mathcal{F}$. For a given filtration $\mathcal{F}_t$, we say that a random field is predictable with respect to $\mathcal{F}_t$, if it is measurable with respect to the $\sigma$-algebra generated by all simple random fields $X_t(h)$ of the form 
\[ 
X_{t}(h) = \sum_{i=1}^{n}\sum_{j=1}^{m}w_{i}Y_{j}\mathbf{1}_{(t_{j},t_{j+1}]}(t)\mathbf{1}_{A_i}(h), \quad (t,h)\in \mathbb{R}\times \Pi, \, A_{i}\in\mathcal{B}(\Pi),
\]
where $w_{i}$ are constants and each $Y_{j}$ is $\mathcal{F}_{t_j}$-measurable. When $L$ is an orthogonal martingale measure with respect to $\mathcal{F}_t$, the class of $L$-integrable random fields are then all $\mathcal{F}_t$-predictable fields $X_{t}(h)$ with $(t,h)\in \mathbb{R}\times \Pi$, such that 
\begin{equation}\label{eq:integrability_condition}
\mathbb{E}\left[ \int_{\mathbb{R}\times \Pi}(X_{s}(\xi)^2 + \lvert X_{s}(\xi)\rvert) (\lambda^{1}\otimes \lambda_{\Pi})(ds,d\xi) \right] < \infty.
\end{equation}
\begin{remark}\label{rem:OMM}
If $\sigma$ is independent of $\mathcal{F}_t^{L}$, then $L$ remains an orthogonal martingale measure in the enlarged filtration $\mathcal{F}_t=\sigma\left( \mathcal{F}_{t}^{L}\cup \mathcal{F}_{t}^\sigma \right)$, where $\mathcal{F}_t^\sigma =\sigma(\lbrace\sigma_{s}(\xi) : s\leq t, \xi \in \Pi\rbrace)$. 
\end{remark}
\begin{definition}[Ambit set]
A collection of sets $\lbrace A_{t}(h)\rbrace_{(t,h)}$ in $\mathcal{M}$ is called a collection of ambit sets if they are causal in the sense that $A_{s}(h) \subseteq A_{t}(h)$ for $s\leq t$ and $A_{t}(h)\cap \left((t,\infty)\times \Pi\right) = \emptyset$.
\end{definition}
\begin{definition}[Ambit field]\label{def:ambit_field}
Let $\lbrace A_{t}(h)\rbrace_{(t,h)}$ be a family of ambit sets in $\mathcal{M}$ and let $L$ be a Lévy basis on $(M,\mathcal{M})$ with intensity measure $c$. Suppose that $L$ is an orthogonal martingale measure with respect to the right continuous and complete filtration $\mathcal{F}_t$. Let $\kappa,K$ be deterministic kernel functions and let $a_{t}(h),\sigma_{t}(h)$ be random fields that are predictable with respect to $\mathcal{F}_t$ and such that the fields $X_{s}^{1}(\xi)=\mathbf{1}_{A_{t}(h)}\kappa (t,s,h,\xi)a_{s}(\xi)$ and $X_{s}^{2}(\xi) = \mathbf{1}_{A_{t}(h)}K(t,s,h,\xi)\sigma_{s}(\xi)$ satisfy \eqref{eq:integrability_condition} for all $(t,h)\in M$. An ambit field $Y_{t}(h)$ on $M$ is then a random field indexed by $(t,h)$ with $t\in\mathbb{R}$ and $h\in \Pi$, defined by the following stochastic integral interpreted in the sense of \textcite{Walsh1986}
\[
Y_{t}(h) = \int_{A_{t}(h)} \kappa (t,s,h,\xi)a_{s}(\xi) c(ds,d\xi) + \int_{A_{t}(h)}K(t,s,h,\xi)\sigma_{s}(\xi) L(ds,d\xi).
\]
\end{definition}
The first term of the ambit field $Y_{t}(h)$ in Definition~\ref{def:ambit_field} represents the usual notion of a drift term, while the second term represents the noise. The ambit sets provide great flexibility, as they allow the modeller to impose a priori known structure into the field in a natural way. Furthermore, one of the powerful properties of ambit fields, is that they allow one to model ``in stationarity'', by potentially including all time points back to the infinite past. In many cases it is natural to consider time intervals of the form $[t_{1},t_{2}]$ and this may easily be accommodated by including an indicator function in the kernel of the form 
$K(t,s,h,\xi)\mathbf{1}_{[t_1,t_2]}(t)$, but in Proposition~\ref{prop:ambit_stationary} we give necessary and sufficient conditions on the kernel function $K$ and volatility field $\sigma$ in the specific model \eqref{eq:model} to ensure stationarity under time translations and spatial rotations, where it is crucial that the integration is carried out over the full history $(-\infty ,t]$. 
\begin{proposition}\label{prop:ambit_stationary}
Let $Y_{t}(h)$ be an ambit field as in Definition \ref{def:ambit_field}, with ambit sets corresponding to the truncated cylinder surface \eqref{eq:ambit_set}. Let the CQ of the driving Lévy basis $L$ be $(\gamma ,\Sigma,\nu ,\lambda_{\mathcal{C}})$, where $\lambda_{\mathcal{C}}$ is the Riemannian measure on $M=\mathcal{C}$. Let also $B\in\mathcal{M}_b$ and define the operations $g_{\tau}$ and $g_{\theta}$ as the following translation and rotation maps on $\mathcal{C}$.
\[
g_{\tau}(B) = B + (\tau,0,0), \quad g_{\theta}(B) = \mathcal{R}_{\theta}B, \quad \tau,\theta \in \mathbb{R},
\]
where $\mathcal{R}_{\theta}$ is the rotation matrix
\begin{equation}\label{eq:rotation_matrix}
\mathcal{R}_{\theta} = \begin{pmatrix}
1 & 0 & 0 \\ 0 & \cos (\theta) & -\sin(\theta) \\ 0 & \sin(\theta) & \cos (\theta) 
\end{pmatrix} = \begin{pmatrix}
1 & \mathbf{0} \\
\mathbf{0} & R_{\theta}
\end{pmatrix}.
\end{equation}
Let $G$ be the group generated by the isometries $g_{\tau},g_{\theta}$, with function composition as group operation. Then $g(B)\in \mathcal{M}_b$ and $L$ is stationary with respect to $G$. Furthermore, the field $Y_{t}(h)$ is stationary with respect to $G$ in the sense that
\[
Y_{t}(h) \overset{d}{=} Y_{t+\tau}(R_{\theta}h), \quad \tau \in \mathbb{R},\; \theta \in \mathbb{R},
\]
if and only if the following conditions hold.
\begin{enumerate}
\item $\sigma_{s}(\xi)$ and $a_{s}(\xi)$ are stationary with respect to $G$, in the sense that
\[
\sigma_{s}(\xi) \overset{d}{=}  \sigma_{s+\tau}(R_{\theta}\xi), \quad a_{s}(\xi) \overset{d}{=} a_{s+\tau}(R_{\theta}\xi).
\]
\item The kernel functions $\kappa,K$ are stationary and isotropic, in the sense that they only depend on the distance between points in as follows (with some abuse of notation)
\[
K(t,s,h,\xi) = K(t-s, \theta_{h}-\theta_{\xi}), \quad \kappa (t,s,h,\xi) = \kappa ( t-s, \theta_{h}-\theta_{\xi}),
\]
where $h=(\cos (\theta_h),\sin(\theta_h))$ and $\xi = (\cos (\theta_{\xi}),\sin(\theta_{\xi}))$, $\theta_h,\theta_{\xi}\in (0,2\pi]$.
\end{enumerate}
\end{proposition}
\begin{corollary}\label{cor:ambit_stationary}
Let the setting be as in Proposition \ref{prop:ambit_stationary}. Then $Y_{t}(h)$ is time stationary such that
$Y_{t}(h)\overset{d}{=}Y_{t+\tau}(h)$,
if and only if the following conditions hold
\begin{enumerate}
\item $\sigma_{s}(\xi)$ and $a_{s}(\xi)$ are time stationary such that
\[
\sigma_{s}(\xi) \overset{d}{=}  \sigma_{s+\tau}(\xi), \quad a_{s}(\xi) \overset{d}{=} a_{s+\tau}(\xi).
\]
\item The kernel functions $K,\kappa$ are of the form (with some abuse of notation)
\begin{equation}\label{eq:stationary_kernel}
K(t,s,h,\xi) = K(t-s, h,\xi), \quad \kappa (t,s,h,\xi) = \kappa (t-s,h,\xi).
\end{equation}
\end{enumerate}
\end{corollary}

\subsection{Model cumulant function and moments}\label{sec:cumulant}
Let $S_{t}(h)$ be the spot price model defined as in \eqref{eq:model}. Then the manifold on which the ambit field is defined is exactly $M=\mathcal{C}$ and we define the \emph{de-seasonalized spot price} as the zero-drift ambit field $D_{t}(h)$, given by
\begin{equation}\label{eq:deseasonalized_field}
\begin{aligned}
D_{t}(h) &= S_{t}(h) - \int_{A_{t}(h)}\kappa (t,s,h,\xi)a_{s}(\xi)c(ds,d\xi) \\
&= \int_{A_{t}(h)}K(t,s,h,\xi)\sigma_{s}(\xi)L(ds,d\xi),
\end{aligned}
\end{equation}
The underlying filtration $\mathcal{F}_t$ is henceforth assumed to be such that $L$ is an orthogonal martingale measure and $\sigma$ is predictable, cf. Definition~\ref{def:ambit_field}. When $\sigma$ and $L$ are independent, we follow Remark~\ref{rem:OMM} and set $\mathcal{F}_{t}=\sigma \left( \mathcal{F}_t^L \cup \mathcal{F}_t^\sigma\right)$, where $\mathcal{F}^L_{t}$ is given by \eqref{eq:L_filtration} and $\mathcal{F}_t^\sigma$ is the filtration
\begin{equation}\label{eq:filtration_sigma}
\mathcal{F}_{t}^{\sigma} = \sigma\left( \lbrace \sigma_{s}(\xi) : (s,\xi)\in (-\infty,t]\times \Pi \rbrace \right).
\end{equation}
Proposition~\ref{prop:ambit_field_cumulant} below states the $\mathcal{F}_{t}^\sigma$-conditional cumulant function of the ambit field \eqref{eq:deseasonalized_field}, in the case where $\sigma$ is independent of $L$. The first and second order structure of the field $D_t(h)$ is then derived in Proposition~\ref{prop:second_order_structure}. When $c=\lambda_{\mathcal{C}}$ is the Riemannian measure, these quantities reduce to volume integrals.
\begin{proposition}\label{prop:ambit_field_cumulant}
Let $D_{t}(h)$ be the deseasonalized field \eqref{eq:deseasonalized_field}, and suppose that the volatility field $\sigma$ is independent of the driving Lévy basis $L$. Then the $\mathcal{F}_{t}^{\sigma}$-conditional cumulant function of $D_{t}(h)$ is given by
\begin{equation}\label{eq:cumulant_function}
\begin{aligned}
C(u;D_{t}(h)\mid \mathcal{F}_{t}^{\sigma}(h)) &= \log\left(\mathbb{E}\left[ e^{\mathrm{i}uD_{t}(h)}\mid \mathcal{F}_{t}^{\sigma}(h) \right]\right)\\
&= \int_{A_{t}(h)} C(uK(t,s,h,\xi)\sigma_{s}(\xi) ;L') c(ds,d\xi) \\
&= \int_{-\infty}^{t}\int_{0}^{2\pi} C(uK(t,s,h,r(\theta))\sigma_{s}(r(\theta)) ;L') d\theta ds.
\end{aligned}
\end{equation}
\end{proposition}
\begin{proposition}\label{prop:second_order_structure}
Let $D_{t}(h)$ be the deseasonalized field \eqref{eq:deseasonalized_field}, and suppose that the volatility field $\sigma$ is independent of the driving Lévy basis $L$.
Then the first and second order structure of $D_{t}(h)$ is given as follows.
\[
\begin{aligned}
\mathbb{E}\left[ D_{t}(h) \right] &= \mathbb{E}\left[ L' \right]\int_{-\infty}^{t}\int_{0}^{2\pi}K(t,s,h,r(\theta))\mathbb{E}\left[ \sigma_{s}(r(\theta))\right] d\theta ds, \\[1.25ex]
\mathrm{Cov}\left( D_{t}(h), D_{t'}(h') \right) &= \mathbb{V}\left[ L'\right]\int_{-\infty}^{\min\lbrace t, t' \rbrace}\int_{0}^{2\pi}K(t,s,h,r(\theta))K(t',s,h',r(\theta))\mathbb{E}\left[ \sigma_{s}(r(\theta ))^{2} \right]d\theta ds \\
+\mathbb{E}\left[ (L') \right]^2\int_{-\infty}^{t}\int_{-\infty}^{t'}&\int_{0}^{2\pi}\int_{0}^{2\pi}K(t,s,h,r(\theta))K(t',s',h',r(\theta '))\varrho (s,s',r(\theta), r(\theta '))d\theta d\theta ' ds ds',
\end{aligned}
\]
where
$
\varrho (s,s',r(\theta), r(\theta')) = \mathrm{Cov}\left( \sigma_{s}(r(\theta)), \sigma_{s'}(r(\theta')) \right).
$
\end{proposition}
\begin{example}\label{ex:BNS_model}
A Barndorff-Nielsen-Shephard type ambit field model is obtained by letting the squared volatility field be an ambit field of the form
\[
\sigma_{t}(h)^2 = \int_{A_{t}(h)}Q(t-s,h,\xi) L^\sigma (ds,d\xi),
\]
where $Q$ is a strictly positive deterministic kernel function satisfying the conditions of Corollary~\ref{cor:ambit_stationary} and $L^\sigma$ is a positive Lévy basis with CQ $(0,0,\nu,\lambda_{\mathcal{C}})$ and independent of $L$. In this case, the volatility field is always non-negative and if the driving noise $L$ is Gaussian with CQ $(0,1,0,\lambda_{\mathcal{C}})$, then the unconditional cumulant function of $D_{t}(h)$ can be computed from Proposition~\ref{prop:ambit_field_cumulant} and the stochastic Fubini theorem of \textcite{Walsh1986} as
\[
C(u;D_{t}(h)) = \int_{A_{t}(h)}C\left( -u^2 I(t-w,h,\zeta); (L^\sigma)'\right)\lambda_{\mathcal{C}}(dw,d\zeta),
\]
where $(L^\sigma)'$ is the Lévy seed of $L^\sigma$ and $I$ is the integral
\[
I(t-w,h,\zeta) = \int_{A_{t}(h)\setminus A_{w}(h)}K(t-s,h,\xi)^2 Q(s-w,\xi,\zeta)\lambda_{\mathcal{C}}(ds,d\xi).
\]
\end{example}
\section{Electricity derivatives}\label{sec:derivatives}
In this Section we consider the pricing of some fundamental electricity derivatives. In the null-spatial case, some semi-explicit pricing formulas for futures and options along with structure preserving changes of measure are derived in \textcites{BARNDORFF-NIELSENOLEE.2013Mesp}{RowińskaPaulinaA.2021Amat}, and we effectively extend these to cover the tempo-spatial setting. The most fundamental product is the futures contract, which is simply the promise of delivery or consumption of a fixed quantity of electricity at a fixed price $P$ in each delivery period over some time period, normalized with the amount of payments. In terms of the payoff, a long futures position on one unit of electricity with delivery period $[\tau_1, \tau_2]$ therefore has the payoff 
\[
\frac{1}{2\pi(\tau_2 - \tau_1)}\int_{\tau_1}^{\tau_2}\int_{0}^{2\pi}S_{t}(r(\phi))d\phi dt - P,
\]
for some fixed strike $P$. As futures contracts are traded products, these are priced under a pricing measure $\mathbb{Q}$ equivalent to the physical measure $\mathbb{P}$. Given such a pricing measure, the futures price at time $\tau_{0}\leq \tau_{1}$, $F_{\tau_{0},\tau_{1},\tau_{2}}$ becomes \citep[see equation~(1.12) in][]{BenthMonograph}
\begin{equation}\label{eq:futures_price}
F_{\tau_{0},\tau_{1},\tau_{2}} = \frac{1}{2\pi(\tau_2 - \tau_1)}\int_{\tau_1}^{\tau_2}\int_{0}^{2\pi}\mathbb{E}^{\mathbb{Q}}\left[ S_{t}(r(\phi)) \mid \mathcal{F}_{\tau_{0}} \right] d\phi dt - P, 
\end{equation}
where $\mathcal{F}_t$ is a filtration representing the available information at time $t$. If the model is such that ${\mathbb{E}^{\mathbb{Q}}[\lvert S_{t}(h)\rvert]<\infty}$, then $F_{\tau_{0},\tau_{1},\tau_{2}}$ is a martingale process in $\tau_{0}$ for any choice of $\mathbb{P}$-equivalent measure $\mathbb{Q}$ (with respect to $\mathcal{F}_{\tau_0}$). This is for example the case in the setting of Proposition~\ref{prop:second_order_structure}, where additionally the volatility field is time stationary and with finite first moment. The futures price in \eqref{eq:futures_price} depends on the \emph{conditional} expectation of $S_{t}(h)$, which is complicated by the path dependence of $t\mapsto S_{t}(h)$. However, as the futures price is simply an expectation, it is straightforward to incorporate the drift term that we have omitted so far, as it can simply be added. For non-linear payoffs with the futures as underlying, such as options, the drift contributes in a non-trivial way to the overall price.

\subsection{Structure preserving change of measure}
Electricity itself is not a storable commodity, and the spot price process is therefore not constrained to be a martingale under $\mathbb{Q}$. It follows that the market is not only incomplete, but that \emph{any} probability measure equivalent to $\mathbb{P}$ is a valid pricing measure. It is therefore common to assume $\mathbb{P}=\mathbb{Q}$, but several works such as \textcites{BenthCarteaKiesel2006}{WeiLunde2023} have found there to be a significant and time-varying risk premium in the electricity market, which necessitates a model flexible enough to capture this. Although the class of valid pricing measures is very large, it is customary to seek a structure preserving change of measure, such that the model of $S_{t}(h)$ does not change its fundamental properties when moving from $\mathbb{P}$ to $\mathbb{Q}$. We may characterize a class of structure preserving changes of measure by means of the Esscher transform, similar to the approach of \textcites{BenthMonograph}{BARNDORFF-NIELSENOLEE.2013Mesp}. We state this result in Proposition \ref{prop:esscher_transform} and for technical reasons, we restrict attention to the case of a finite time interval $[0,T^{\ast}]$, meaning that for the rest of the Section, we interpret the sets $A_{t}(h)$ as starting at time $t=0$. The main implication is, that time stationarity of the field $S_{t}(h)$ is lost c.f. Corollary \ref{cor:ambit_stationary}.
\begin{proposition}\label{prop:esscher_transform}
Let $L$ be a Lévy basis on $[0,T^\ast]\times \Pi$ with CQ $(\gamma, \Sigma ,\nu, \lambda_{\mathcal{C}})$, and let $(\mathcal{F}_t^L)_{t\in [0,T^{\ast}]}$ be the filtration generated by $L$, corresponding to \eqref{eq:L_filtration}. Let also $q(t,h):[0,T^\ast]\times \Pi\to \mathbb{R}$ be a continuous function and suppose that there exists a $q^{*}$ such that
\[
\mathbb{E}\left[ e^{q^{*}\lvert L([0,T^\ast]\times \Pi)\rvert} \right]<\infty, \quad \sup_{(t,h)\in [0,T^{\ast}]\times \Pi}\lvert q (t,h)\rvert < q^{*}.
\]
Define the process $Z_t$ as
\[
Z_t = \exp\left({\int_{0}^{t}\int_{\Pi}q (s,\xi) L(d s,d\xi)-\int_{0}^{t}\int_{\Pi}C(q (s,\xi) ;L')\lambda_{\mathcal{C}}(d s,d\xi)}\right), \quad t\in [0,T^{\ast}].
\]
Then $Z_t=\frac{\mathbb{P}^{q}}{\mathbb{P}}\big\rvert_{\mathcal{F}_t}$ defines a valid density process with respect to the filtration $(\mathcal{F}_{t}^L)_{t\in [0,T^{\ast}]}$, and it holds that $L$ is a Lévy basis under the measure $\mathbb{P}^q$ (defined on $(\Omega,\mathcal{F}\rvert_{T^{\ast}})$) with inhomogeneous CQ $(\tilde{\gamma},\Sigma,\tilde{\nu},\lambda_{\mathcal{C}})$, where
\begin{equation}\label{eq:esscher_triplet}
\begin{aligned}
\tilde{\gamma}(t,h) &= \gamma + q(t,h)\Sigma + \int_{\mathbb{R}}\left( e^{q(t,h)x}-1 \right)\mathbf{1}_{\lvert x\rvert < 1}(x) \nu (d x), \\
\tilde{\nu}(t,h,dx) &= e^{q(t,h)x}\nu (d x).
\end{aligned}
\end{equation}
\end{proposition}
Proposition \ref{prop:esscher_transform} imposes the constraint that $L$ must admit sufficiently many exponential moments, which is not a significant constraint, but it precludes the use of heavy-tailed Lévy bases. This change of measure is -- as usual -- essentially an altering of the drift of the driving noise, but with added flexibility due to the jump structure. We may consider $q (t,h) > 0$ to emphasize the occurrence of positive jumps, and $q (t,h) <0$ to emphasize the occurrence of negative jumps. The sign of $q$ therefore determines which type of risk the market prices into the futures contracts, and we may use $q (t,h)$ to construct estimates of the market price of risk. Finally we note that a similar change of measure may be carried out for the volatility field $\sigma_{t}(h)$, which can be used to characterize a market price of volatility risk.

\subsection{Electricity futures} 
The futures contract constitutes the most fundamental electricity derivative, but also acts as the underlying of other derivative products such as options. As remarked above, the futures price is a conditional expectation, which is generally not expressible in terms of the current price level as the model is not Markovian. Suppose for simplicity that we have fixed some pricing measure $\mathbb{Q}$ via the structure preserving Esscher transformation with $\tilde{q}(t,h)=q$ constant, that $\sigma$ is independent of $L$, and that $\mathcal{F}_t$ is the filtration $\mathcal{F}_t=\sigma\left( \mathcal{F}_t^L \cup \mathcal{F}_{t}^{\sigma} \right)$ of Remark~\ref{rem:OMM}. Applying  Proposition~\ref{prop:second_order_structure} and a simple computation using the independently scattered property of $L$ shows that
\begin{equation}\label{eq:futures_price2}
\begin{aligned}
\mathbb{E}^{\mathbb{Q}}\left[ D_{T}(h) \mid \mathcal{F}_{t} \right] &= \int_{A_{t}(h)}K(T,s,h,\xi)\sigma_{s}(\xi) L(d s, d \xi) \\
&\quad + \mathbb{E}^{\mathbb{Q}}\left[ L'\right]\int_{t}^{T}\int_{0}^{2\pi}K(T,s,h,r(\theta))\mathbb{E}^{\mathbb{Q}}\left[ \sigma_{s}(r(\theta))\mid \mathcal{F}_{t}\right]d\theta d s.
\end{aligned}
\end{equation}
Up to the computation of the expectation $\mathbb{E}^{\mathbb{Q}}\left[ \sigma_{s}(r(\theta)) \mid\mathcal{F}_{t} \right]$, this is reasonably explicit, but the first term on the right hand side is generally not equal to $D_{t}(h)$. This may be imposed by assuming that $K(T,t,h,\xi)=K(T,h,\xi)K(t,h,\xi)$, but this is also a quite restrictive assumption. Following Chapter 10.3.3 of \textcite{Ambit}, we may construct the squared volatility field $\sigma^{2}$ (which is positive) to allow for semi-closed forms of the conditional expectation via the functional identity
\[
x^{p} = \frac{p}{\Gamma (1-p)}\int_{0}^{\infty}\frac{(1-e^{-xz})}{z^{1+p}}d z,
\]
in which case we find that
\[
\mathbb{E}^{\mathbb{Q}}\left[ \sigma_{s}(\xi)\mid \mathcal{F}_{t} \right] = \mathbb{E}^{\mathbb{Q}}\left[ \sqrt{\sigma_{s}(\xi)^2} \mid \mathcal{F}_{t} \right] = \frac{1}{2\sqrt{\pi}}\int_{0}^{\infty}\frac{1-\mathbb{E}^{\mathbb{Q}}\left[ e^{-z\sigma_{s}(\xi)^{2}}\mid\mathcal{F}_t\right]}{z^{\frac{3}{2}}}d z.
\]
By specifying $\sigma^{2}$ such that its conditional Laplace transform under $\mathbb{Q}$ is known in closed form, we may then plug this into \eqref{eq:futures_price2} and integrate to obtain the final futures price. This can for example be computed via \eqref{eq:cumulant_function} in the specification of Example~\ref{ex:BNS_model}. Ultimately, we note that the futures price is essentially identical to the null-spatial case of \textcites{BARNDORFF-NIELSENOLEE.2013Mesp}{Ambit} up to the additional spatial integral. Hence the extension to the tempo-spatial setting is a rather smooth generalization, and the cost incurred from adding the extra dimension is not too severe. 

As noted in \textcite{AidMonograph}, conventional Markovian spot price models have a hard time reproducing the observed volatility structure of futures prices at even slightly long maturities or delivery periods, as the implied exponential mean-reversion speed essentially entails the spot price process staying near its long term mean over the interval $[\tau_1, \tau_2]$. The present framework, however, allows for much more flexible autocorrelation structures, and thus more realistic volatility patterns, depending on the specification of $K$ and $\sigma$. We will return to the issue of selecting the kernel function in Section~\ref{sec:kernel}.

In the literature, it is customary to think of an electricity futures contract as being composed by a sum of forward contracts, which has delivery at an instant in time. If we define the forward price as $\tilde{F}_{t}(T,h) = \mathbb{E}^{\mathbb{Q}}\left[ S_{T}(h) \mid \mathcal{F}_t \right]$, then it holds that
\[
F_{t,\tau_1,\tau_2} = \frac{1}{2\pi (\tau_2-\tau_1)}\int_{\tau_1}^{\tau_2}\int_{0}^{2\pi}\tilde{F}_{t}(T,r(\theta)) d T d\theta.
\]
Both the forward and futures contracts are martingales in time, and hence we may follow \textcite{BARNDORFF-NIELSENOLEE.2013Mesp} and derive their $t$-dynamics $d\tilde{F}_{t}(T,h),d F_{t,\tau_1,\tau_2}$, which may provide some insight on the structure and volatility of the futures price. Due to tempo-spatial setting, however, we will not recover a usual stochastic differential equation, and it is not immediately clear how to obtain useful dynamics for the forward price. Assuming for simplicity that the drift term is zero, we may apply the stochastic Fubini theorem in \textcite{Walsh1986} to obtain the representation
\[
\begin{aligned}
\tilde{F}_{t}(T,h) &= \int_{0}^{t}\int_{\Pi}K(T,s,h,\xi)\sigma_{s}(\xi)L(d s,d\xi) \\
&\quad + \mathbb{E}^{\mathbb{Q}}\left[ L'\right]\int_{t}^{T}\int_{0}^{2\pi}K(T,s,h,r(\phi))\mathbb{E}^{\mathbb{Q}}\left[ \sigma_{s}(r(\phi)) \mid \mathcal{F}_t \right]d\phi d s,
\end{aligned}
\]
Since $\sigma$ is independent of $L$, Proposition~\ref{prop:second_order_structure} yields that
\begin{equation}\label{eq:samuelson_effect}
\lim_{t\uparrow T}\mathbb{V}^{\mathbb{Q}}\left[ \tilde{F}_{t}(T,h) \right] = \mathbb{V}^{\mathbb{Q}}\left[ D_{T}(h) \right].
\end{equation}
This simply shows that the volatility of the forward price converges to that of the (deseasonalized in this case) spot price. In classical models, this convergence is exponentially increasing, which is known as the Samuelson effect. In our framework, however,  we can obtain different convergence speeds, depending on the choice of kernel $K$. We may obtain similar expressions for both spatially averaged forward contracts and futures contracts, where we note that the spatially averaged forward contract
$
\hat{F}_{t}(T) = \int_{0}^{2\pi}\tilde{F}_{t}\left( T,r(\theta)\right) d\theta,
$
corresponds to the usual notion of a forward contract in the null-spatial setting. The spatially averaged forward contract has the following structure
\[
\begin{aligned}
\hat{F}_{t}(T) &= \int_{0}^{t}\int_{\Pi}I(T,s,\xi)\sigma_{s}(\xi) L(d s,d\xi) \\
&\quad +  \mathbb{E}^{\mathbb{Q}}\left[ L'\right]\int_{t}^{T}\int_{0}^{2\pi}I(T,s,r(\phi))\mathbb{E}^{\mathbb{Q}}\left[ \sigma_{s}(r(\phi))\mid \mathcal{F}_t\right]d\phi d s,
\end{aligned}
\]
where $I(T,s,\xi)=\int_{0}^{2\pi}K(T,s,r(\theta),\xi)d\theta$ is the spatially integrated kernel. This corresponds to a type of volatility modulated mixed moving average process as studied in \textcite{Veraart2015}, and the results therein may be useful to derive the autocorrelation function and finite dimensional distribution of the (spatially averaged) forward price. The convergence of the volatility structure corresponding to \eqref{eq:samuelson_effect} is now in an averaged sense of the form
\[
\lim_{t\uparrow T}\mathbb{V}^{\mathbb{Q}}\left[\hat{F}_{t}(T)\right] = \mathbb{V}^{\mathbb{Q}}\left[ \int_{0}^{2\pi}D_{T}(r(\theta))d\theta \right].
\]
For the spatially averaged forward, the kernel $I(T,t,\xi)$ plays the role of the kernel function in \textcite{BARNDORFF-NIELSENOLEE.2013Mesp} and the behaviour of $t\mapsto I(T,t,\xi)$ will determine how the forward volatility converges to the average spot volatility.

In addition to plain futures contracts, there is a variety of more exotic derivatives in electricity markets. It is clear that the pricing on hedging of more such products is not free of complications and the tractability of the pricing formulas strongly depend on the model specification. We note, however, that it is possible to simulate the ambit model as detailed in Section \ref{sec:simulation}, and thereby to price certain products via Monte Carlo methods without the need for explicit analytical expressions.

\subsection{Within-day spreads}
Our model distinguishes itself by including each individual delivery period in the extra spatial dimension, and it is therefore possible to consider a class of derivatives which have, to the best of our knowledge, not been studied previously in the literature. As a particular example, we consider products written on the average within-day spread between selected delivery periods. Such products do not currently trade on any centralized exchange, but may potentially be used for hedging delivery commitments, or to gain exposure to spike risk in certain delivery periods. Concretely, for  given $H_1,H_2\in \mathcal{B}((0,2\pi])$ we consider the spread
\[
X_{\tau_{1},\tau_{2}} = \int_{\tau_1}^{\tau_2}\left[\int_{H_1}S_{t}(r(\phi))d\phi - \int_{H_2}S_{t}(r(\phi))d\phi\right] dt.
\]
The analysis of derivatives written on the spread $X_{\tau_1,\tau_2}$ turns out to have a relatively simple structure in our setting, as the stochastic Fubini theorem of \textcite{Walsh1986} yields a mixed moving average representation similar to that of the futures price
\begin{equation}\label{eq:spread_payoff}
X_{\tau_1,\tau_2} = \int_{\tau_1}^{\tau_2}\left[ \int_{A_{t}(h)}I_{1}(t,s,\xi)a_{s}(\xi)\lambda_{\mathcal{C}}(ds,d\xi) + \int_{A_{t}(h)}I_{2}(t,s,\xi)\sigma_{s}(\xi)L(ds,d\xi) \right] dt,
\end{equation}
where $I_{1}$ and $I_{2}$ denote the following quantities
\[
\begin{aligned}
I_{1}(t,s,\xi) &= \int_{H_1}\kappa (t,s,r(\phi),\xi) d\phi - \int_{H_2}\kappa (t,s,r(\phi),\xi) d\phi, \\
I_{2}(t,s,\xi) &= \int_{H_1}K(t,s,r(\phi),\xi)d\phi - \int_{H_2}K(t,s,r(\phi),\xi)d\phi.
\end{aligned}
\] 
The structure \eqref{eq:spread_payoff} is obviously retained under structure preserving changes of measure, and hence the kernel functions are what drives the prices of derivatives written on $X_{\tau_1,\tau_2}$, which highlights the importance of modelling these to properly account for the risk inherent in the spread $X_{\tau_{1},\tau_{2}}$. The value of the spread is simply obtained as
\begin{equation}\label{eq:withi_day_spread_price}
\begin{aligned}
\mathbb{E}^{\mathbb{Q}}\left[ X_{\tau_1,\tau_2} \mid \mathcal{F}_{\tau_0}\right] &= 
\int_{\tau_1}^{\tau_2}\left[\int_{0}^{\tau_0} \int_{0}^{2\pi}I_{1}(t,s,r(\phi))a_{s}(r(\phi)) d\phi ds + \int_{A_{\tau_0}(0)}I_{2}(t,s,\xi)\sigma_{s}(\xi)L(d\xi, ds) \right] dt \\
&\quad + \int_{\tau_1}^{\tau_2}\left[\int_{\tau_0}^{t} \int_{0}^{2\pi}I_{1}(t,s,r(\phi))\mathbb{E}^{\mathbb{Q}}\left[ a_{s}(r(\phi))\mid \mathcal{F}_{\tau_0} \right] d\phi ds\right] dt \\
&\quad + \mathbb{E}^{\mathbb{Q}}\left[ L' \right] \int_{\tau_1}^{\tau_2}\left[\int_{\tau_{0}}^{t}\int_{0}^{2\pi}I_{2}(t,s,r(\phi)) \mathbb{E}^{\mathbb{Q}}\left[ \sigma_{s}(r(\phi))\mid\mathcal{F}_{\tau_0} \right] d\phi ds\right] dt.
\end{aligned}
\end{equation}
Similar to the futures price, the first term in \eqref{eq:withi_day_spread_price} is generally not equal to any directly observable quantity. An example of a within-day spread of particular interest could be the case where $H_{1} = \left[ \frac{2\pi}{3}, \frac{5\pi}{3}\right]$, corresponding to the daily \emph{peak load period} and $H_{2}=(0,2\pi]$ or $H_{2}=(0,\frac{2\pi}{3})\cup (\frac{5\pi}{3},2\pi]$ corresponding to the average daily price or the \emph{off-peak load period}. A spread on these subsets of delivery periods would provide exposure to deviations in the peak load period from the overall daily average, or the discrepancy between prices in peak and off-peak periods. The case of $H_{2}=(0,2\pi]$ is particularly convenient, as the integrals $I_1$ and $I_2$ simplify to
\begin{equation}\label{eq:kernel_integrals}
I_{1}(t,s,\xi) = -\int_{(0,\frac{2\pi}{3})\cup (\frac{5\pi}{3},2\pi]}\kappa (t,s,r(\phi),\xi) d\phi, \quad
I_{2}(t,s,\xi) = -\int_{(0,\frac{2\pi}{3})\cup (\frac{5\pi}{3},2\pi]}K(t,s,r(\phi),\xi)d\phi,
\end{equation}
which corresponds to the price of the spread being proportional to the expected average of the off-peak load period. Consistent modelling of within-day spread contracts constitutes an interesting area of research in itself, but a deeper analysis of the pricing of within-day spreads and their potential for hedging spike risk is left for future research. 

\subsection{Incorporating leverage}
In equity markets, the effect of \emph{leverage} is well established. Leverage is the simple notion that large upwards movements in volatility are typically associated with downwards movements in the price. As noted in \textcite{InverseLeverage}, this effect is seemingly opposite in electricity markets, such that periods of high volatility are associated with high price levels. In our modelling framework, there are multiple ways of adding leverage effects and it is in particular possible to let the volatility field $\sigma$ and the driving noise $L$ be dependent in \eqref{eq:model}. Since many results simplify immensely under the assumption of independence between $\sigma$ and $L$, we illustrate a straightforward way of incorporating leverage in the spirit of \textcite{BNS}, which maintains independence of $\sigma$ and $L$. If we assume that the variance field $\sigma_{t}(h)^{2}$ is specified as in Example~\ref{ex:BNS_model}, the leverage effect can be modelled as
\begin{equation}\label{eq:leverage}
D_{t}(h) = \int_{A_{t}(h)}K(t-s,h,\xi)\sigma_{s}(\xi) L(d s , d\xi) + \int_{A_{t}(h)}\rho(s)\widetilde{K}(t-s,h,\xi) \widetilde{L}^{\sigma}(d s, d\xi),
\end{equation}
where $\rho (t)$ is a function modelling the ``amount of leverage'', $\widetilde{K}$ is an admissible kernel function and 
\[
\widetilde{L}^{\sigma}(A) = L^{\sigma}(A), \quad \text{or} \quad \widetilde{L}^{\sigma}(A) = L^{\sigma}(A) - \mathbb{E}\left[ L^{\sigma}(A) \right].
\]   
The first case simply means that volatility directly contributes to the price level, whereas the second case is such that only periods of excess volatility (compared to the mean) drives up the price. If the leverage effect is interpreted not on the level of shocks, but rather trends, we may also model leverage by taking $a_{s}(\xi) = f(\sigma_{s}(\xi))$ in \eqref{eq:model} for some appropriate increasing function $f$. In this view the volatility field directly modulates the differentiable drift, such that a period of high volatility contributes smoothly to a rise in the overall price level, which gradually goes back down when the level of volatility does.

While the inclusion of leverage effects is simple to model, it complicates the resulting estimation and pricing formulas even more than just including stochastic volatility. The impact on futures prices will be linear when $\tilde{L}^{\sigma}=L^{\sigma}$ and negligible when $\tilde{L}^{\sigma}$ has mean zero, but the impact on prices of nonlinear payoffs will generally be much more apparent, as the variance of the process $D_{t}(h)$ substantially increases for $\rho (t) \geq 0$, due to the positive correlation of the two terms in \eqref{eq:leverage}. It is clear that understanding the finer structure of spot price volatility is necessary to assess the impact and form of leverage. 

\section{Specifying the kernel function}\label{sec:kernel}
The kernel function $K$ entering in the model \eqref{eq:model} is a crucial modelling ingredient that determines the dependence structure in the model. Ideally, $K$ should be chosen to reflect the various stylized facts of electricity prices, but still be tractable enough to facilitate the computation of required model quantities, such as moments and derivative prices. It is well-documented that properly de-seasonalized electricity prices are stationary in time, and in all of the following we therefore impose the conditions of Corollary~\ref{cor:ambit_stationary} on the kernel, entailing that $K$ is of convolution type in time. We note that all model computations simplify immensely whenever $K$ is \emph{separable} in the sense that
\begin{equation}\label{eq:separable_kernel}
K(t-s,h,\xi) = K_{1}(t-s)K_{2}(h,\xi),
\end{equation}
as all quantities involving $K$ essentially split into a product of a temporal and a spatial part. Unfortunately, it is likely more realistic that the temporal and spatial parts interact in a non-linear way, prompting the need for a non-separable kernel. This issue is related to that of non-separable covariance functions, which is a well-known problem within random field theory and is e.g. discussed in \textcite{CressieNoel1999CoNS}.

\subsection{A semi-parametric approximation scheme}\label{sec:kernel_approximation}
Selecting a kernel with a known functional form is necessary to do any meaningful computations within the model, but is prone to model misspecification and/or intractable model expressions. As an alternative to selecting a functional form, we now illustrate a semi-parametric approach, which approximates any given time-stationary model of the form \eqref{eq:deseasonalized_field} via generic parameters to be fitted (or \emph{learned}) from the data. This is very useful for model estimation, where one assumes that there exists a correctly specified and time stationary model $S_{t}(h)$ of the form \eqref{eq:model} for the true data generating process, with \emph{unknown} and possibly non-separable kernel $K$. By taking the approximation order sufficiently large, we may then estimate or calibrate the semi-parametric model on the data to uncover the true shape of $K$ without imposing any structure. This can be achieved by noting that the square integrable kernel function $K\in L^{2}(\mathbb{R}_{+}\times \Pi \times \Pi)$ of \eqref{eq:deseasonalized_field} can be represented as
\begin{equation}\label{eq:repr_K_L2}
K(t,r(\theta),r(\phi)) = \frac{1}{2\pi}\sum_{j_1=0}^\infty\sum_{j_{2}\in\mathbb{Z}}\sum_{j_{3}\in\mathbb{Z}}c_{j_{1},j_{2},j_{3}}\Psi_{j_1}(t)e^{\mathrm{i}j_{2}\theta}e^{\mathrm{i}j_{3}\phi},
\end{equation}
where $\Psi_{j_1}(t)$ is given by
\[
\Psi_{j_1}(t) = \left(\frac{j_1 !}{\Gamma (j_1 + \alpha + 1)}\right)^{1/2}t^{\tfrac{\alpha}{2}}e^{-\frac{t}{2}}\mathcal{L}_{j_1}^{(\alpha)}(t),
\]
with $\mathcal{L}_{j_1}^\alpha$ denoting the $j_{1}$'th generalized Laguerre polynomial with parameter $\alpha > -1$. The representation \eqref{eq:repr_K_L2} follows since the system $(\Psi_{j_1}(t))_{j_1 = 0}^{\infty}$ is an orthonormal basis for $L^{2}(\mathbb{R}_{+})$ (see, e.g., Theorem~5.7.1 in \textcite{szego75}) and $(\frac{1}{2\pi}e^{\mathrm{i}j_{2}\theta}e^{\mathrm{i}j_{3}\phi})_{(j_{1},j_{2})\in \mathbb{Z}\times \mathbb{Z}}$ is an orthonormal basis for $L^2(\Pi\times \Pi)$. The coefficients $c_{j_1,j_2,j_3}$ are the projection coefficients
\[
c_{j_1,j_2,j_3} = \langle K, \Psi_{j_1}e^{\mathrm{i}j_2 \cdot}e^{\mathrm{i}j_3 \cdot}\rangle_{L^2(\mathbb{R}_{+}\times \Pi\times \Pi)}.
\]
Truncating the sum in \eqref{eq:repr_K_L2} yields a finite order approximation $K^{(n_{1},n_{2},n_{3})}$ of the form
\begin{equation}\label{eq:truncated_kernel}
K^{(n_1,n_2,n_3)}(t,r(\theta),r(\phi)) =\frac{1}{2\pi} \sum_{j_1=0}^{n_1}\sum_{j_{2}\in \lbrace -n_{2},\ldots ,n_{2} \rbrace}\sum_{j_{3}\in \lbrace -n_{3},\ldots ,n_{3} \rbrace}c_{j_{1},j_{2},j_{3}}\Psi_{j_1}(t)e^{\mathrm{i}j_{2}\theta}e^{\mathrm{i}j_{3}\phi},
\end{equation}
which is exactly the projection of $K$ onto the finite dimensional subspace
\begin{equation}\label{eq:span}
\mathrm{span}\left( \Psi_{j_1}(t)e^{\mathrm{i}j_2\theta}e^{\mathrm{i}j_3\phi} : 0\leq j_1 \leq n_1, \; \lvert j_2\rvert \leq n_2, \; \lvert j_3 \rvert \leq n_3 \right). 
\end{equation}
We can then define the \emph{approximate field} $D_{t}^{(\mathbf{n})}(h)$ by
\begin{equation}\label{eq:approx_field} 
D_{t}^{(\mathbf{n})}(h) = \int_{A_{t}(h)}K^{(\mathbf{n})}(t-s,h,\xi)\sigma_{s}(\xi)L(ds,d\xi),
\end{equation}
with $\mathbf{n}=(n_1,n_2,n_3)$. The mean-squared error between $D_{t}^{(\mathbf{n})}(h)$ and $D_t(h)$ can be derived via the Itô isometry of \textcite{Walsh1986} as
\begin{equation}\label{eq:ambit_field_isometry_approx}
\begin{aligned}
\mathbb{E}\left[ \left( D_{t}(h) - D_{t}^{(\mathbf{n})}(h)\right)^2\right] &= \int_{A_{t}(h)}\left( K(t-s,h,r(\phi))-K^{(\mathbf{n})}(t-s,h,r(\phi)) \right)^2 d\phi ds \\
&= \int_{-\infty}^{t}\int_{0}^{2\pi}\left(\frac{1}{2\pi}\sum_{(j_1,j_2,j_3)\notin\mathcal{J}}c_{j_{1},j_{2},j_{3}}\Psi_{j_1}(t)e^{\mathrm{i}j_{2}\theta}e^{\mathrm{i}j_{3}\phi}\right)^2 d\theta ds,
\end{aligned}
\end{equation}
where $\mathcal{J}=\lbrace j_1\leq n_1\rbrace\cup\lbrace \lvert j_2\rvert\leq n_2\rbrace\cup \lbrace \lvert j_3\rvert \leq n_3\rbrace$, and it follows that $D_{t}^{(\mathbf{n})}(h)\to D_{t}(h)$ in $L^2(\Omega)$. As the basis in which we expand $K$ is orthonormal, it follows from Parseval's identity that the coefficients $c_{j_1,j_2,j_3}$ tend to zero in the limit, and it is therefore reasonable to expect the approximation error \eqref{eq:ambit_field_isometry_approx} to become small for sufficiently large approximation orders. The finite order accuracy depends on the decay rate of the $c_{j_1,j_2,j_3}$ coefficients, which may potentially be controlled by imposing suitable assumptions on the true kernel $K$. 

Ultimately, we may carry out inference on the true and unknown model $D_t(h)$, by fitting the simpler approximate model $D_t^{(n_1,n_2,n_3)}(h)$ to data. To this end, the approximate model will, without further structure, have $N=(n_1+1)(2n_2+1)(2n_3+1)$ coefficients $c_{j_1,j_2,j_3}$ to be fitted, along with the extra parameter $\alpha$, for a total of $N+1$ parameters. The total amount of parameters quickly blows up in practice, but since the projection coefficients are ``generic'', it may be possible to induce sparsity via various common regularization methods. The parameter $\alpha$ is of independent interest, as it controls the behaviour of the approximate kernel near $t=0$. In particular, we note that for $\alpha\in (-1,0]$, the temporal behaviour is similar to that of the gamma kernel, which is utilized in the context of Lévy semistationary processes for energy markets in \textcite{BARNDORFF-NIELSENOLEE.2013Mesp}{Bennedsen2017}. In the univariate case, the gamma kernel has been used to model \emph{roughness}, in the sense that the price paths are less regular than that of a Brownian motion, and the model is a semimartingale exactly when $\alpha = 0$. With non-Gaussian noise, the notion of Hölder regularity becomes redundant, but we note that even in our tempo-spatial setting, the case of $\alpha=0$ corresponds to a type of semimartingale property in time of the field $D_{t}^{(n_1,n_2,n_3)}(h)$ c.f. Chapter~5.3.3 of \textcite{Ambit} and this semimartingale property is lost for $\alpha < 0$ where the kernel becomes singular such that $\lim_{t\to 0}K(t,h,\xi)=\infty$. We note that even though the model $D_{t}^{(n_1,n_2,n_3)}(h)$ is an approximation, it is itself a very general and flexible model that can model complex dependence structures even at low approximation orders and exhibits similar temporal behaviour to the univariate models of \textcites{BARNDORFF-NIELSENOLEE.2013Mesp}{Bennedsen2017}.

\subsection{An empirical toy study of the kernel function}\label{sec:kernel_fitting}
In Section~\ref{sec:kernel_approximation}, we introduced the model $D_{t}^{(\mathbf{n})}(h)$, where $\mathbf{n}=(n_1,n_2,n_3)$ is a multi-index. The kernel function governing this field is an approximation to the underlying ``true'' kernel, with the added benefit that the approximate kernel consists of basis functions which are simple to deal with. In the case of rotational invariance/isotropy of the field $D_{t}(h)$, we may derive a tractable representation that permits estimation of $D_{t}^{(\mathbf{n})}$ via Whittle likelihood methods that involve the model spectral density which is derived in Proposition~\ref{prop:spectral_density} and requires that the driving Lévy basis has mean zero, i.e., $\mathbb{E}[L']=0$. This is not an unreasonable assumption for the properly de-seasonalized prices as illustrated on German data in Appendix~\ref{app:data}. The Whittle likelihood procedure is, in the general non-Gaussian case, a pseudo-likelihood method that is only able to identify the field up to the second order structure. Loosely speaking, the estimated $K$ can be considered as the kernel of the best Gaussian approximation to the non-Gaussian field, but the parameters of the driving Lévy basis $L$ must be identified by other means. We leave a full estimation procedure for future research and consider this exercise a first step towards characterizing the kernel function.
\begin{proposition}\label{prop:D_spectral_rep}
Let $D_t(h)$ be the de-seasonalized field \eqref{eq:deseasonalized_field} and suppose that the driving Lévy basis $L$ has mean zero, i.e., $\mathbb{E}[L']=0$. Then it holds that
\[
D_{t}(h) \overset{L^2}{=} \frac{1}{2\pi}\sum_{n\in\mathbb{Z}}X_n(t) e^{\mathrm{i}nh},
\]
where $X_n(t)$ are the scalar-valued and stationary (on $\mathbb{R}$) processes
\begin{equation}\label{eq:spectral_processes}
X_{n}(t) = \int_{0}^{2\pi}D_{t}(r(\theta))e^{-\mathrm{i}n\theta}d\theta.
\end{equation}
\end{proposition}

\begin{proposition}\label{prop:spectral_density}
Let $D_{t}(h)$ be as in Proposition~\ref{prop:D_spectral_rep}. Suppose in addition that $D_{t}(h)$ satisfies the conditions of Proposition~\ref{prop:ambit_stationary}, such that $D_{t+\tau}(R_{\theta}h)\overset{d}{=}D_{t}(h)$ with $R_\theta$ denoting the rotation matrix in \eqref{eq:rotation_matrix} and denote by $\mathbb{E}[\sigma^2]$ the stationary second moment of the volatility field $\sigma_{t}(h)$. Then the spectral density of $D_{t}(h)$ is given by $\sum_{n\in\mathbb{Z}}f_{n}$, where each $f_n$ is the spectral density of the process $X_{n}(t)$ in \eqref{eq:spectral_processes} and is given by
\[
f_{n}(u) = \frac{1}{2\pi}\int_{\mathbb{R}}e^{-\mathrm{i}u\tau}\mathbb{E}[X_{n}(t+\tau)\overline{X_n(t)}]d\tau = \frac{\mathbb{V}[L']\mathbb{E}[\sigma^2]}{2\pi} \times \lvert H_{n}(u)\rvert^2,  \quad H_{n}(v) = \sum_{j=1}^{\infty}a_{j,n}\widehat{\Psi}_{j}(v),
\]
where $u\in\mathbb{R}$, $a_{j,n}=c_{j,-n,n}$ with $c_{j_1,j_2,j_3}$ defined as in \eqref{eq:repr_K_L2}, and 
\begin{equation}\label{eq:Psi_fourier_transform}
\widehat{\Psi}_{j}(u) = \left( \frac{j!}{\Gamma (j+\alpha + 1)}\right)^{\frac{1}{2}}\Gamma (1+\tfrac{\alpha}{2})(\tfrac{1}{2}+\mathrm{i}u)^{-(\frac{\alpha}{2}+1)}{}_{2}F_{1}(-j,\tfrac{\alpha}{2}+1;\alpha+1;1-\tfrac{1}{1/2+\mathrm{i}u}),
\end{equation}
with ${}_{2}F_{1}$ denoting the hypergeometric function.
\end{proposition}
Using Proposition~\ref{prop:spectral_density}, we can compute the empirical spectral density from observations of the multivariate time series $(X_{-N},\ldots,X_{0}(t),\ldots ,X_{N}(t))^{\top}$ for some finite $N$ and match this to the model-implied spectral density of Proposition~\ref{prop:spectral_density}. In practice, the transfer functions $H_{n}(v)$ must be truncated at some finite order $J$, which ultimately corresponds exactly to modelling $D_t(h)$ by the approximate (and isotropic) field $D_{t}^{(J,N,N)}(h)$ of Section~\ref{sec:kernel_approximation}. It is therefore reasonable to assume that taking the truncation orders $J,N$ sufficiently large yields a good proxy for the true and unknown kernel function $K$. Since all observation points are equispaced in both time and space, the empirical spectral density can be computed very efficiently via the fast Fourier transform as the periodogram over the time grid $\lbrace t_{1}<t_{2}<\cdots <t_{J}\rbrace$
\begin{equation}\label{eq:empirical_spectral_density}
\widehat{f}_{n}(u_{k}) = \frac{1}{2\pi J}\left\lvert \sum_{j=1}^{J}X_{n}(t_j) e^{-\mathrm{i}u_{k}t_{j}}\right\rvert^2.
\end{equation}
The Whittle estimator is then found as $\vartheta^\ast = \arg\min_\vartheta \ell_W(\vartheta)$, where $\vartheta$ denotes the parameter vector of the model spectral density, $f_n(u;\vartheta)$, and $\ell_W$ is the loss function
\begin{equation}\label{eq:whittle_loss}
\ell_{W}(\vartheta) = \sum_{n=-N}^{N}\sum_{u_{k}}\left[ \log (f_{n}(u_k;\vartheta)) + \frac{\widehat{f}_n(u_k)}{f_{n}(u_k;\vartheta)}\right].
\end{equation}
As is customary in the literature on spectral methods, we assume that the spectral density of $D_t(h)$ satisfies the minimum-phase condition, which in particular entails that $\log (f_{n}(u))$ is integrable and that the transfer function $H_{n}(u)$ is analytic and without zeros in the half-plane $\Im (u)<0$. See, e.g., \textcites{Bingham1}{Bingham2} for surveys on this and related topics. In Proposition~\ref{prop:identifiability} below, it is shown that under the minimum-phase condition plus some additional regularity, the model with a kernel of the form \eqref{eq:repr_K_L2} is well-specified, in the sense that two differently parametrized kernels cannot produce the same model spectral density. This also holds for the approximate kernel $K^{(J,N,N)}$, such that, in practice, we obtain the best possible approximation to the true kernel in the finite dimensional space \eqref{eq:span}. The required additional regularity are mainly technical assumptions which are simple to enforce in a numerical optimization procedure.
\begin{proposition}\label{prop:identifiability}
Let the setting be as in Proposition~\ref{prop:spectral_density}. Suppose in addition that $a_{0,0}=1$ and that for each $n$, it holds that $a_{j^\ast (n),n}>0$, where $j^\ast (n)$ is defined as the first index for which $a_{j^\ast (n),n}\neq 0$. Let $K_1,K_2$ be kernel functions satisfying the conditions of Proposition~\ref{prop:ambit_stationary} and let $(f_{n}^{(1)})_{n\in\mathbb{Z}},(f_{n}^{(2)})_{n\in\mathbb{Z}}$ be their respective spectral densities. If $f_{n}^{(1)}(u)=f_{n}^{(2)}(u)$ for all $n$ and $u$, then $K_1\overset{L^2}{=}K_2$. Additionally, defining $f_{n,J}$ as the spectral density of $X_{n}(t)$ in \eqref{eq:spectral_processes} with the transfer function $H_{n}$ truncated at level $J$, it holds that if $f_{n,J}^{(1)}(u)=f_{n,J}^{(2)}(u)$ for all $u$, $n=-N,\ldots ,N$ and with $J$ fixed, then $K_{1}^{(J,N,N)}\overset{L^2}{=}K_{2}^{(J,N,N)}$.
\end{proposition}
Consistency and asymptotic normality of the Whittle estimator $\vartheta^\ast$ as $J\to\infty$ in \eqref{eq:empirical_spectral_density} holds under standard regularity assumptions on the multivariate process $(X_{-N}(t),...,X_{0}(t),...,X_{N}(t))$ such as existence of fourth moments and weak mixing conditions (see, e.g., \textcite{Hannan1970}). These model  properties are inherited from the stationary field $D_{t}(h)$, where we note that \textcite{CuratoStelzerStroh2022} derive central limit theorems and certain weak mixing conditions for a class of ambit fields on $\mathbb{R}^d$ that are conceptually very similar to that of Example~\ref{ex:BNS_model} (see in particular their Corollaries~4.6 and 4.8). Extending their results to cover our non-Euclidean setting is, however, beyond the scope of this paper.

\begin{figure}
	\begin{center}
	\begin{subfigure}[b]{0.48\textwidth}
		\includegraphics[scale=0.375]{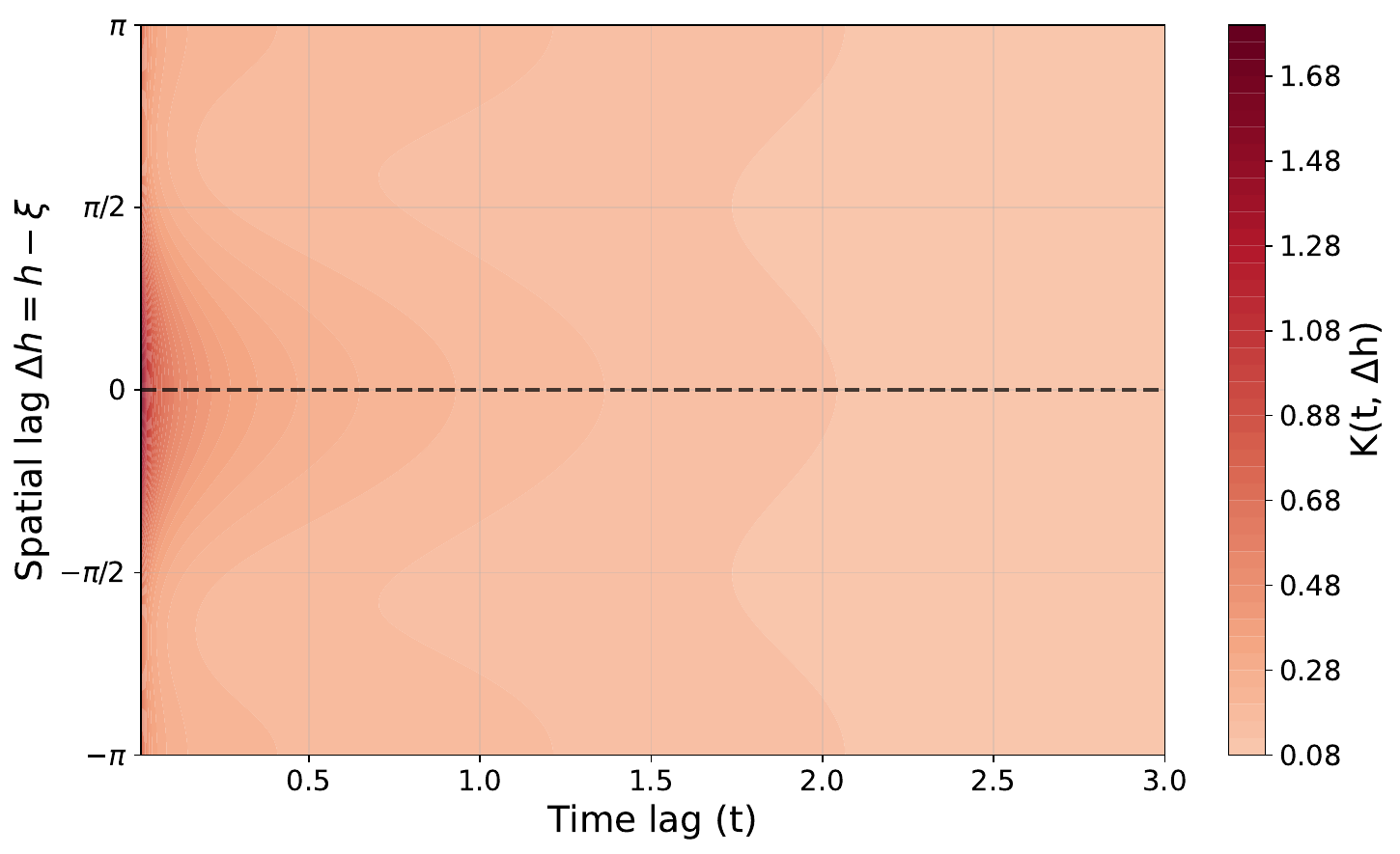}
	\end{subfigure}
	\begin{subfigure}[b]{0.48\textwidth}
		\includegraphics[scale=0.375]{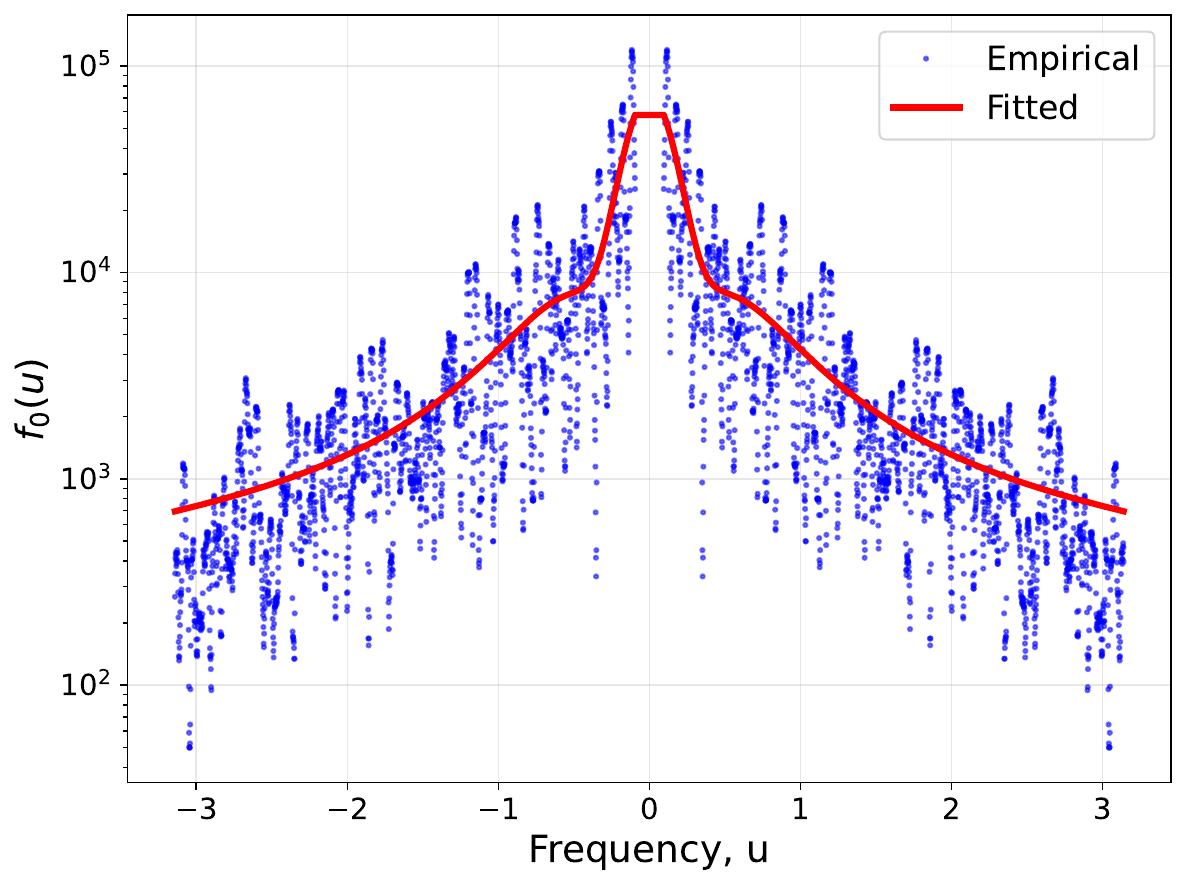}
	\end{subfigure}
	\end{center}
	\caption{Fitted isotropic kernel function (left) and fitted spectral density at mode $n=0$ against the empirical counterpart (right).}\label{fig:kernel_fit}
\end{figure}
As an illustration, we consider minimizing the loss function \eqref{eq:whittle_loss} for the isotropic de-seasonalized model of Proposition~\ref{prop:spectral_density} on de-trended German data from the period October 1st 2018 to October 1st 2024. A description of the de-seasonalization procedure is provided in Appendix~\ref{app:data}. We choose approximation orders corresponding to the field $D_{t}^{(4,2,2)}$ which has parameter vector $\vartheta = (\alpha, \varsigma \, ;\, (a_{j,n}))$, where $\varsigma = \mathbb{V}[L']\mathbb{E}[\sigma^2]$. This leaves 16 total parameters to be estimated, of which 14 are the unrestricted generic coefficients $a_{j,n}$ (for comparison, note that the 24-dimensional ambit process of \textcite{VeraartVeraart2014} has $\approx 200$ parameters). The model kernel function is $K^{(4,2,2)}$ given by \eqref{eq:truncated_kernel} and so it follows from Proposition~\ref{prop:identifiability} that we only require summing over $n\in \lbrace -2,\ldots, 2\rbrace$ in the loss function \eqref{eq:whittle_loss}. To improve the numerical implementation, we utilize a profile likelihood approach where the variance parameter $\varsigma$ is estimated from the first order condition of \eqref{eq:whittle_loss} as
$
\widehat{\varsigma}(\vartheta) = \frac{1}{(2N+1)J}\sum_{n}\sum_{j}\lvert H_{n}(u_j;\vartheta)\rvert^2 f_n(u_{j};\vartheta).
$
The estimated kernel parameters are provided in Table~\ref{tab:kernel_coefficients} and the resulting isotropic kernel function is depicted on Figure~\ref{fig:kernel_fit} along with the fitted spectral mode at $n=0$. The fits of the remaining spectral modes $f_{n}(u)$ for $n\neq 0$ are shown in Appendix~\ref{app:plot}. We note that the empirical spectral densities are very noisy, but that the overall shapes are relatively well-matched by the model spectral densities. The fitted kernel exhibits quite slow decay in the temporal dimension, suggesting that de-seasonalized prices exhibit substantial autocorrelation at long lags. The overall shape of the fitted kernel aligns well with the notions of adjacency, autocorrelation, and cyclicality as discussed in Section~\ref{sec:model_setup}. We also note that the $\alpha$ parameter is essentially estimated to be at the lower bound of $-1$, suggesting that singular behaviour in the kernel function is helpful to reproduce the observed autocorrelation structure of prices. This is consistent with the findings of \textcite{Bennedsen2017} for the daily average price series. The fitted kernel showcases some important features of the dependence structure in spot prices, but suffers from the assumption of rotational invariance. An interesting avenue for further research is to develop a robust estimation methodology for general kernels of the form $K(t-s,h,\xi)$, to shed more light on the cross sectional dependence structure.
\begin{table}[htbp]
\centering
\begin{tabular}{cccccc}
$n$\textbackslash $j$ & 0 & 1 & 2 & 3 & 4 \\
\midrule
0 & $1.0000$ & $-0.1939$ & $0.0609$ & $-0.0713$ & $0.0305$ \\
1 & $0.5608$ & $0.0293$ & $0.0117$ & $0.0025$ & $-0.0014$ \\
2 & $0.2490$ & $-0.0031$ & $0.0054$ & $-0.0063$ & $-0.0003$ \\
\midrule
\multicolumn{6}{c}{$\alpha = -0.9910$, \quad $\varsigma = 6.4677 \times 10^{4}$} \\
\bottomrule
\end{tabular}
\vspace{0.2cm}
\caption{Estimated kernel parameters $a_{j,n}=c_{j,-n,n}$ based on Whittle likelihood. The coefficient $a_{0,0}$ is fixed at $1$ to ensure the minimum-phase condition.}
\label{tab:kernel_coefficients}
\end{table}

\section{A simulation scheme}\label{sec:simulation}
In this Section, we develop a simulation scheme for an ambit field of the form \eqref{eq:deseasonalized_field}, which is inspired by the methodology developed in \textcites{BenthEyjolfssonVeraart2014}{Eyjolfsson2015}, but with the advantage that the proposed scheme is able to treat singular kernels without having to truncate the kernel at some finite level. The procedure relies on a representation of the ambit field as a countable sum of certain integrated processes, which ultimately permits an iterative simulation scheme. For ambit fields on $\mathbb{R}^d$, this type of representation is generally only an approximation, which is obtained after truncating the spatial domain to a compact subset. In our setting, however, the spatial domain is already compact, and hence the representation of the ambit field as a countable sum is in fact exact and given in Proposition \ref{prop:simulation}. The simulation scheme is based on truncating this sum at some finite level, and in Proposition \ref{prop:error_bound} we provide a bound on the error between the true and the truncated ambit field. To obtain the desired representation, we assume that time starts at $t=0$ and make the following technical assumption on the kernel function.
\begin{assumption}\label{ass:laplace_kernel}
Let $K(t-s,h,\xi)$ be a kernel function satisfying the criteria of Corollary \ref{cor:ambit_stationary}. Then we assume that the function $t\mapsto K(t,h,\xi)$ is continuously differentiable on $(0,\infty)$ and such that $\lvert K(t,h,\xi) \rvert < Me^{-\gamma t}$ for some $M > 0$ and $\gamma > 0$.
\end{assumption}
\begin{proposition}\label{prop:simulation}
Let $Y_{t}(h)$ be an ambit field according to Definition \ref{def:ambit_field} of the form
\[
Y_t(h) = \int_{0}^{t}\int_{\Pi}K(t-s,h,\xi)\sigma_{s}(\xi)L(ds,d\xi),
\]
where we assume that $\mathbb{E}[L']=0$ and $K$ and $\sigma$ satisfy the conditions of Corollary \ref{cor:ambit_stationary} with $\sigma$ independent of $L$, and  $\Pi$ is the unit circle. 
Then $Y_{t}(h)$ has the representation
\begin{equation}\label{eq:simulation_representation}
Y_{t}(h) = \sum_{n\in\mathbb{Z}}\int_{\mathbb{R}}\widehat{K}(z_{r}+\mathrm{i}z_{i},h,n) V_{n}(t,z_{r}+\mathrm{i}z_{i}) dz_{i},
\end{equation}
where $z_{r}\in (-\gamma,0)$ and
\begin{equation}\label{eq:kernel_transforms}
\begin{gathered}
\widehat{K}(z,h,n) = \int_{0}^{\infty}e^{-zx}\tilde{K}(x,h,n) dx, \quad
\widetilde{K}(t-s,h,n) = \frac{1}{2\pi}\int_{0}^{2\pi}e^{\mathrm{i}\phi n}K(t-s,h,r(\phi)) d\phi,
\end{gathered}
\end{equation}
and $t\mapsto V_{n}(t,z)$ is a complex-valued OU type process defined by
\begin{equation}\label{eq:Vn_OU}
V_{n}(t,z) = \int_{0}^{t}\int_{\Pi}e^{(z_{r}+\mathrm{i}z_{i})(t-s)}e^{\mathrm{i}\xi n}\sigma_{s}(\xi) L(ds,d\xi).
\end{equation}
\end{proposition}
\begin{figure}[htbp]
	\centering
	\begin{subfigure}[b]{0.495\textwidth}
		\includegraphics[scale=0.4]{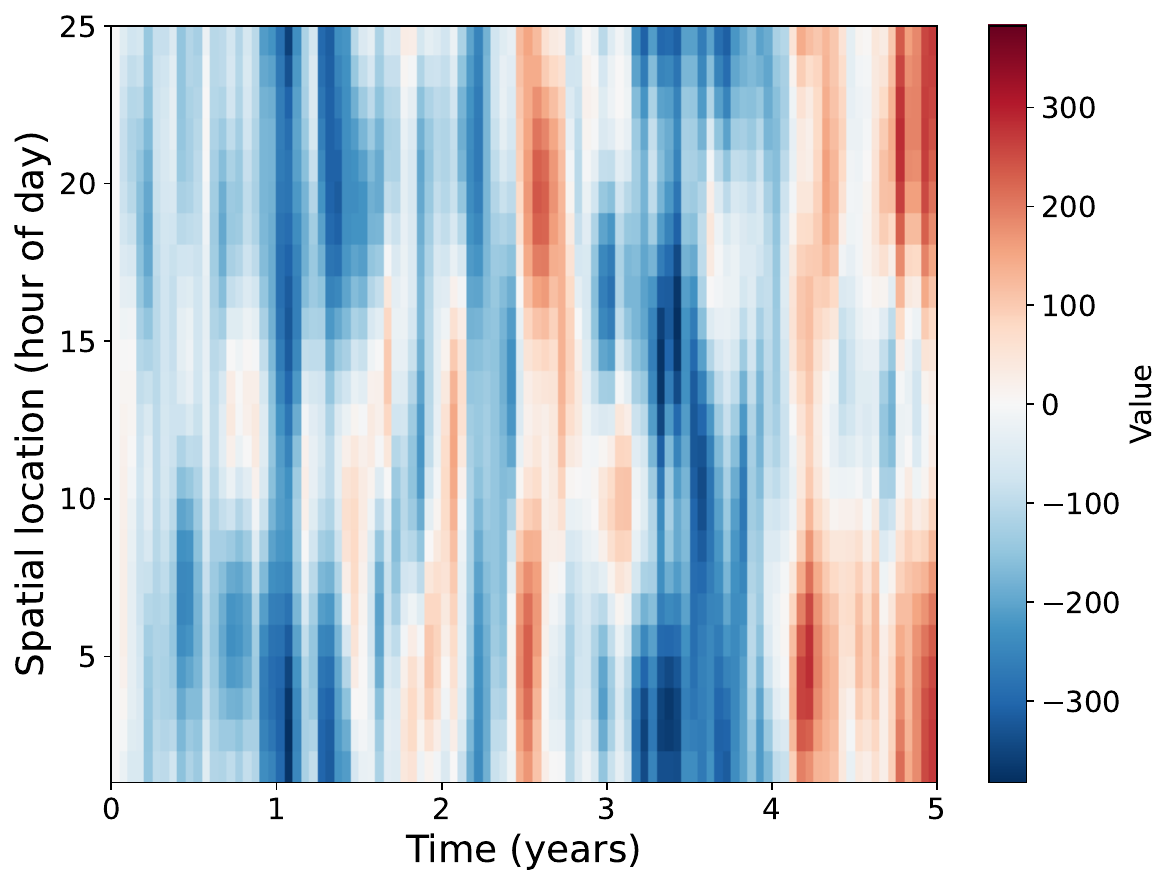}
	\end{subfigure}
	\begin{subfigure}[b]{0.495\textwidth}
		\includegraphics[scale=0.4]{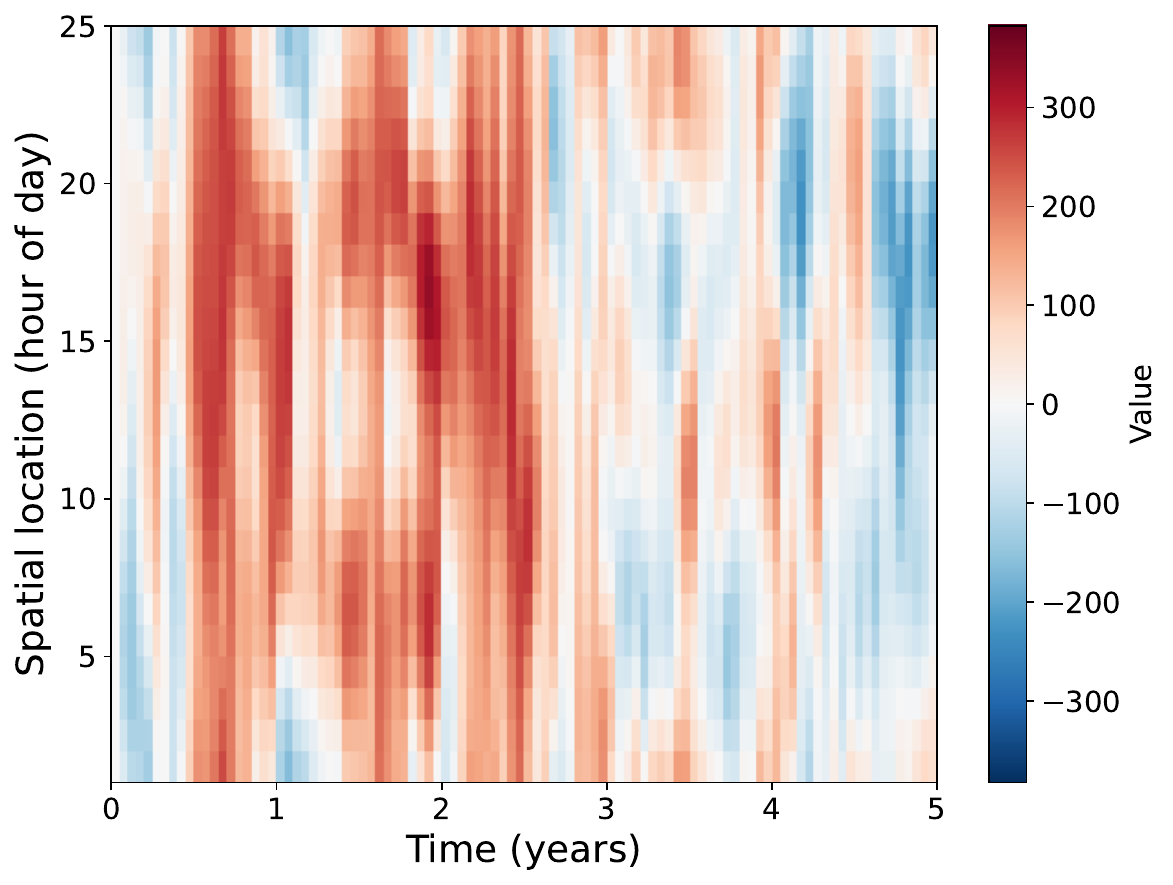}
	\end{subfigure}
	\caption{Generic realizations of simulated Gaussian (left) and NIG (right) ambit fields without stochastic volatility.}\label{fig:sim_fields}
\end{figure}
To see why the representation \eqref{eq:simulation_representation} of Proposition \ref{prop:simulation} is useful, we truncate the countable sum at some finite $N$ to obtain
\begin{equation}\label{eq:Yn_approximation}
Y_{t}(h) \approx Y_{t}^{(N)}(h) = \sum_{n=-N}^{N}\int_{\mathbb{R}}\widehat{K}(z_r + \mathrm{i}z_i, h, n)V_{n}(t,z) dz.
\end{equation}
This allows for an iterative (in time) simulation scheme of the approximate field since we may simulate the OU type process $t\mapsto V_{n}(t,z)$ iteratively in time, via an Euler scheme as
\begin{equation}\label{eq:sim_complex_OU}
V_{n}(t_{k+1},z) = e^{(z_r + \mathrm{i}z_{i})(t_{k+1}-t_{k})}V_{n}(t_{k},z) + e^{(z_r + \mathrm{i}z_{i})(t_{k+1}-t_{k})}\left( L_{n}(t_{k+1})-L_{n}(t_{k})\right),
\end{equation}
where
\[
L_{n}(t) = \int_{0}^{t}\int_{\Pi}e^{\mathrm{i}n \xi}\sigma_{s}(\xi) L(ds,d\xi).
\]
If we can simulate the volatility field $\sigma_t (h)$, we can then discretize the set $[t_{k},t_{k+1}]\times \Pi=\bigcup_{j=1}^{J}\Pi_j$ and simulate the stochastic integral by simulating from $L(\Pi_j)$. The simulation procedure for an ambit field without stochastic volatility modulation is detailed in Example~\ref{ex:simulation_algorithm} in Appendix~\ref{app:simulation}, where the general case with stochastic volatility modulation follows by first generating the volatility field according to Example \ref{ex:simulation_algorithm}. Simulation code for reproducing the plots can be found at \textcolor{blue}{\url{github.com/Tkkloster/ambit_simulation}}. Proposition \ref{prop:error_bound} below provides an estimate on the error between the exact and the simulated field, where we can control the approximation error if we can control the decay rate of the coefficients $\lvert \hat{K}(z,h,n) \rvert$ in $n$. Furthermore, we note that the approximation error is reduced when $\lvert z_r\rvert$ is large. We would therefore like to choose $z_r$ close to $-\gamma$, but for numerical stability this is not always desirable. In our experience, selecting $z_{r}=-\gamma/2$ usually works well.
\begin{figure}[htbp]
	\centering 
	\includegraphics[scale=0.4]{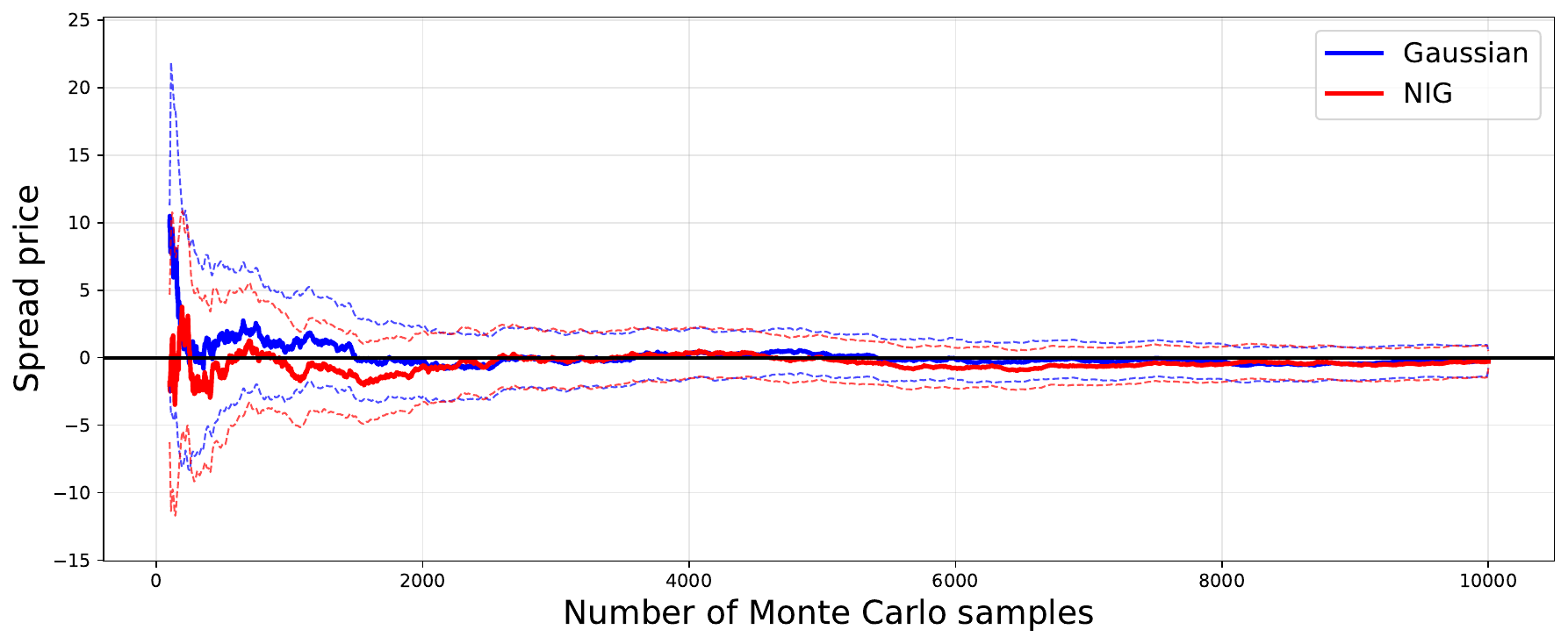}
	\caption{Price of within-day spread as computed via Monte Carlo simulation, with the increasing sample size on the horizontal axis. The Gaussian model is represented as a blue line, and the NIG model is represented as a red line. Corresponding $95\%$ confidence intervals are indicated as dotted lines.}\label{fig:spread_simulation}
\end{figure}
\begin{proposition}\label{prop:error_bound}
Let $Y_t$ be defined as in Proposition \ref{prop:simulation} and $Y_t^{(N)}$ as in \eqref{eq:Yn_approximation}. Suppose that there is a $k_{\sigma}>0$ such that $\mathbb{E}[\sigma_{t}(h)^2]<k_{\sigma}$ for all $(t,h)\in [0,T]\times \Pi$, and that $z_{r}<0$. Then we have the following bound on the approximation error
\begin{equation}\label{eq:simulation_bound}
\mathbb{E}\left[ \left( Y_{t}(h)-Y_{t}^{(N)}(h)\right)^2 \right] \leq \frac{\mathbb{V}[L'] k_{\sigma}}{8\pi^2 \lvert z_r\rvert}\cdot \sum_{n\in\mathbb{Z}\setminus \lbrace -N,\ldots ,N \rbrace}\left(\int_{\mathbb{R}}\lvert \widehat{K}(z_r+\mathrm{i}z_i,h,n)\rvert dz_{i}\right)^2.
\end{equation}
\end{proposition}
\begin{remark}
In the case of a kernel of the semi-parametric form \eqref{eq:truncated_kernel}, the Fourier transform $\widetilde{K}(t-s,h,n)$ is non-zero for only finitely many $n$ by definition, so we can simulate without any error incurred by truncation of the countable sum. 
\end{remark}

\subsection{A concrete example}
We now illustrate the proposed simulation scheme by simulating from a model with the fitted kernel of Section~\ref{sec:kernel_fitting} with parameters as in Table~\ref{tab:kernel_coefficients}, which should then exhibit a realistic covariance structure. We thus simulate from a rotation invariant model (without stochastic volatility for simplicity) of the form
\[
Y_{t}(h) = \int_{0}^{t}\int_{\Pi}K^{(4,2,2)}(t-s,\theta_h-\theta_\xi)L(ds,d\xi),
\]
We consider both the case where $L$ is a Gaussian Lévy basis as in Example~\ref{ex:Gaussian_levy_basis} and where $L$ is a NIG Lévy basis as in Example~\ref{ex:NIG_levy_basis}. By inspection of the kernel \eqref{eq:truncated_kernel}, we find that it satisfies Assumption~\ref{ass:laplace_kernel} with $\gamma = \frac{1}{2}$ and we therefore take the integration contour to be $z_{r}=-\frac{1}{4}$. The transform $\widehat{K}$ is simply given by $\widehat{K}(z,h,n)=e^{\mathrm{i}nh}H_{n}(z)$ with $H_{n}(z)$ as in Proposition~\ref{prop:spectral_density}. The variance of the Gaussian noise is chosen to reflect the estimated variance $\mathbb{V}[L']$ in Section~\ref{sec:kernel_fitting} at $\Sigma=64,677$. For the NIG Lévy basis, we choose the parameters as
\[
\delta_{NIG} = 710, \quad \alpha_{NIG} = 0.5, \quad \beta_{NIG} = 0.48, \quad \mu_{NIG} = -\frac{\delta_{NIG}\beta_{NIG}}{\sqrt{\alpha_{NIG} - \beta_{NIG}}},
\]
which are such that the noise has zero mean and with the same variance as the Gaussian noise. The parameters $\alpha_{NIG},\beta_{NIG}$ are then tuned such that the skewness and excess kurtosis of the NIG Lévy basis closer to their unconditional empirical counterparts for the observed spot price residuals.

In Figure \ref{fig:sim_fields}, we provide an example of a generic realization of each ambit field. For the simulation, we consider a time interval of $T=5$ years, and simply set $N=2$ which encompasses all of the spatial information in the kernel. We simulate the field $Y_{t}(h)$ at 24 spatial points over weekly time increments using a spatial discretization of $72$. In the notation of Example~\ref{ex:simulation_algorithm}, this means $J=260$ and $H=72$. 
\begin{figure}[htbp]
    \centering
    \begin{subfigure}[b]{0.48\linewidth}
        \centering
        \includegraphics[scale=0.3]{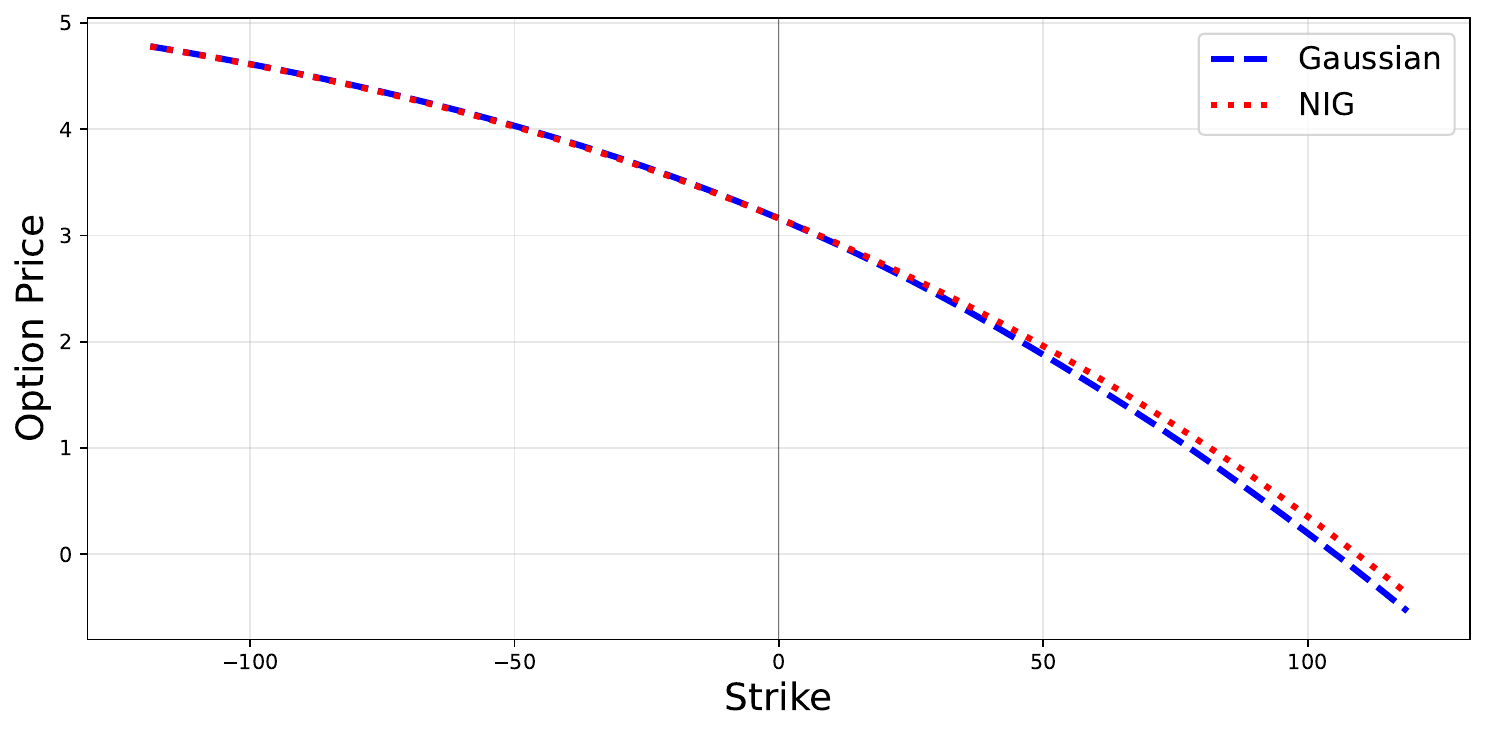}
    \end{subfigure}
    \hfill
    \begin{subfigure}[b]{0.48\linewidth}
        \centering
        \includegraphics[scale=0.3]{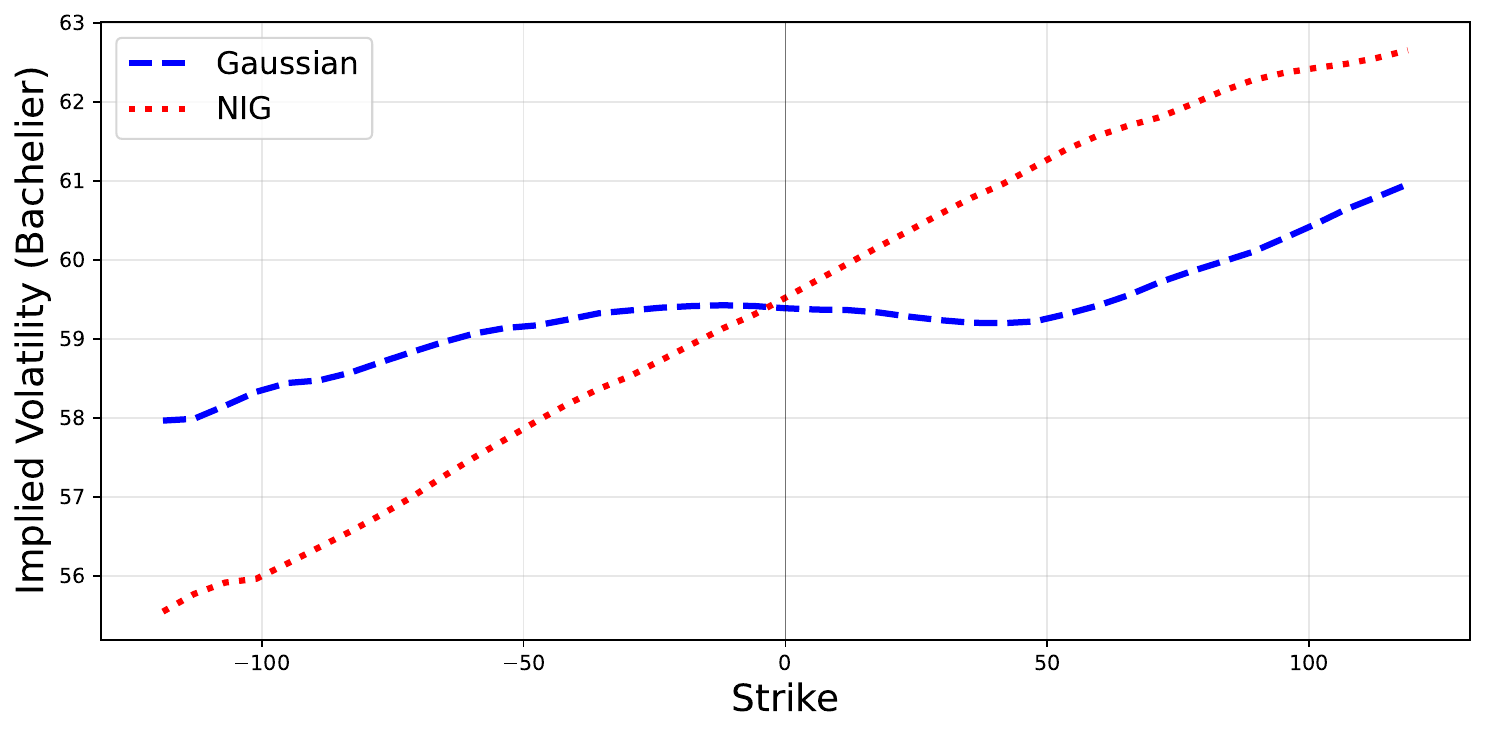}
    \end{subfigure}
    \caption{Spread option prices (left) and Bachelier implied volatilities (right) in the Gaussian (blue dashed line) and NIG (red dotted line) models across a range of strikes from -0.05 to 0.05.}
    \label{fig:spread_option_price}
\end{figure}

As an example of derivatives pricing, we consider the problem of pricing within-day spreads as well as options on these. Recall that the within-day spread has payoff $X_{\tau_1,\tau_{2}}$ given by \eqref{eq:spread_payoff}, and a call option then has payoff $\left( X_{\tau_1,\tau_2} - P \right)^{+}$, for some fixed strike $P$. For simplicity, we suppose that the model is already specified under a pricing measure $\mathbb{Q}$. We specify that the within-day spread contract has a year long settlement period from $t=1$ to $t=2$, and that the spread is between peak load and off-peak load prices during this period. In the notation of Section \ref{sec:derivatives}, we have $\tau_{0}=0,\tau_1 = 1,\tau_2 =2$ and $H_{1}=[\frac{2\pi}{3},\frac{5\pi}{3}]$, $H_{2}=(0,\frac{2\pi}{3})\cup (\frac{5\pi}{3},2\pi]$. From \eqref{eq:withi_day_spread_price}, we know that the fair price of the within-day spread is zero, and in Figure \ref{fig:spread_simulation} we illustrate the accuracy of the Monte Carlo scheme across 10,000 simulations, where we see the simulated price steadily converging towards the true fair price. An option on this spread introduces non-linearity in the payoff and we are thus unable to compute these prices analytically, but on Figure \ref{fig:spread_option_price}, we illustrate the Monte Carlo call option price across a range of strikes for 10,000 simulations. The two models produce visually indistinguishable option prices, so we also compare a normalized price in terms of the associated Bachelier implied volatilities, corresponding to the implied volatility of a call option in the Bachelier model. This does not correspond directly to the forward starting nature of the within-day spread option under consideration, but still serves as a magnifying lens for the option prices. It is apparent that although the two models produce seemingly similar prices, they generate distinct shapes of the implied volatility curves, clearly showing the non-zero skewness of the NIG model. 

\section{Conclusion}\label{sec:conclusion}
We have introduced a continuous time tempo-spatial modelling framework for the full panel of electricity spot prices, which allows for risk management on the most granular scale possible in a consistent manner. The framework is based on a tempo-spatial ambit field indexed by a manifold, which allows for a parsimonious extension of the null-spatial framework developed in \textcite{BARNDORFF-NIELSENOLEE.2013Mesp}. By specifying the ambit sets as cylinder surfaces in $\mathbb{R}^3$, we have shown how to obtain a time stationary model with several key dependence structures intrinsically embedded, such as the circular nature of the delivery periods. To set up the framework, we developed some new results regarding ambit fields on manifolds, which presents a novel application of ambit stochastics. The developed framework is very general and encompasses many previously proposed models for the average spot price, and we have discussed how to model the key stylized facts of electricity spot prices, much of which hinges on the chosen kernel function. We have also treated the pricing of derivatives, where several results from the null-spatial setting of \textcite{BARNDORFF-NIELSENOLEE.2013Mesp} have been extended to the tempo-spatial one, thereby permitting the treatment of derivatives down to the individual delivery period. In particular, we have derived a structure preserving change of measure and futures pricing formulas, which are essentially as tractable as in the null-spatial case and briefly introduced a potentially interesting derivative product, based on within-day spreads between delivery periods. The price of such a spread is mainly determined by the dependence structure in the model, which highlights the importance of modelling this appropriately via the kernel function. As a first application of the developed framework, we have therefore introduced a semi-parametric kernel that can be fitted to data via pseudo-likelihood methods. Applying this to German data reveals that a kernel function with flexible structure in both time and cross section is required and we find that a temporal singularity is strongly suggested by the data, which is in line with previous findings in \textcite{Bennedsen2017}. This framework could be used to guide the development of consistent and robust parametric model specifications in the future.

To further facilitate pricing and estimation in the model, we have also developed a simulation scheme and illustrated how to implement this. In future work, this may be helpful to assess the quality of statistical estimators in finite samples and the pricing of complicated products. A deeper study of the pricing and hedging of derivative products is particularly interesting as the possibility of treating products down to the individual delivery period allows for much more flexible risk management. In this regard, the within-day spreads present a unique instrument for hedging cross sectional risks, which, to the best of our knowledge, have not previously been treated in the literature.

\section*{Acknowledgments} 
The author gratefully acknowledges Fred Espen Benth, Mikkel Bennedsen, and Almut Veraart for their careful proofreading and insightful comments, which have no doubt improved the paper's presentation. Additionally, the author wishes to thank Peter K. Christensen, Elisa Nicolato and Jan Pedersen for helpful comments, as well as the participants at the XXVI Quantitative Finance Workshop in Palermo, the 12\textsuperscript{th} general AMaMeF conference in Verona, and the 2025 Vienna Congress on Mathematical Finance. Financial support has been granted by Center of Research in Energy: Economics and Markets and The Danish Council of Independent Research under DFF grant 10.46540/5247-00005B.

\section*{Conflicts of interest}
The author declares that there are no relevant financial or non-financial competing interests to report.
  
\printbibliography

\appendix
\section{Proofs}\label{app:proofs}
\paragraph{Proof of Proposition \ref{prop:levy_basis_stationary}}
\begin{proof}
Suppose first that $c$ is invariant with respect to $G$. We need to show that 
\[
(L(A_{1}),L(A_{2}),\ldots ,L(A_{n})) \overset{d}{=} (L(g(A_1)),L(g(A_2)),\ldots ,L(g(A_n))),
\] 
for any finite collection $A_1,A_2,\ldots ,A_n$ of elements of $\mathcal{M}_{b}$, such that $g(A_1),g(A_2),\ldots ,g(A_n)$ are in $\mathcal{M}_b$. Note that whenever $g(A)\in\mathcal{M}_b$, the random measure $L(g(A))$ is still the same Lévy basis $L$ on the same space, $(M,\mathcal{M})$, with the same CQ such that 
\begin{equation}\label{eq:g(A)_cumulant}
C(u ; L(g(A))) = C(u;L')c(g(A)).
\end{equation}
Following the proof of Proposition 32 in \textcite{Ambit}, there exists a partition of say $N_{n}\in\mathbb{N}$ mutually disjoint sets $g(B_{l})$ such that
\[
\bigcup_{j=1}^{n}g(A_{j}) = \bigcup_{l=1}^{N_n}g(B_{j}).
\]
Define the indices $I_j\subseteq \lbrace 1,\ldots ,N_n \rbrace$ for $j=1,\ldots , n$ such that 
\[
g(A_j) = \bigcup_{k\in I_{j}}g(B_l).
\]
As the $g(B_l)$'s are disjoint, we have that
\[
L(g(A_j)) = \sum_{l\in I_j}L(g(B_l)).
\]
The characteristic function of $(L(g(A_1)), \ldots ,L(g(A_n)))$ for $u\in \mathbb{R}^n$ is thus 
\[
\begin{aligned}
\mathbb{E}\left[ \exp\left( \mathrm{i}\sum_{j=1}^{n}u_{j}L(g(A_j)) \right)\right] &= \prod_{l=1}^{N_n}\mathbb{E}\left[ \exp\left( \mathrm{i}\cdot L(g(B_l))\sum_{j=1}^{k}u_{j}\mathbf{1}_{I_{j}}(l)\right) \right] \\
&= \prod_{l=1}^{N_n}\exp\left( C\left( \sum_{j=1}^{k}u_{j}\mathbf{1}_{I_j}(l)\; ; \; L(g(B_l)) \right) \right),
\end{aligned}
\]
where $C(x;X)$ denotes the cumulant function of the random variable $X$ evaluated in the point $x$. From \eqref{eq:g(A)_cumulant} we find that $C(u;L(g(B_l))) = c(g(B_l))C(u;L')$ and it therefore holds that
\[
\begin{aligned}
\mathbb{E}\left[ \exp\left( \mathrm{i}\sum_{j=1}^{n}u_{j}L(g(A_j)) \right)\right] &= \prod_{l=1}^{N_n}\exp\left( c(g(B_l))C\left( \sum_{j=1}^{k}u_{j}\mathbf{1}_{I_j}(l) ; L' \right) \right) \\
&= \prod_{l=1}^{N_n}\exp\left( c(B_l)C\left( \sum_{j=1}^{k}u_{j}\mathbf{1}_{I_j}(l) ; L' \right) \right) \\
&= \mathbb{E}\left[ \exp\left( \mathrm{i}\sum_{j=1}^{k}u_{j}L(A_{j}) \right) \right],
\end{aligned}
\]
which shows that $L$ is stationary with respect to $G$. On the other hand, if $L$ is stationary with respect to $G$, it holds that
\[
C(u;L(g(A))) = C(u;L(A)), \quad A\in\mathcal{M}_b,
\]
which in turn implies
\[
C(u;L(g(A)) = c(g(A))C(u;L') = C(u;L(A)) = c(A)C(u;L').
\]
As this holds for arbitrary $g\in G$ and $A\in\mathcal{M}_b$, it follows that $c(A) = c(g(A))$.
\end{proof}
\paragraph{Proof of Proposition \ref{prop:pushforward}}
\begin{proof}
By Proposition 2.6 in \textcite{Rosinski1989}, the cumulant function of $L\circ \varphi^{-1}$ is for $A\in \mathcal{M}_{b}$ given by
\[
\begin{aligned}
C\left(u;L\left(\varphi^{-1}(A)\right)\right) &= \log\left(\mathbb{E}\left[ e^{uL(\varphi^{-1}(A))}\right]\right) \\
&=C(u;L')\lambda^{2}(\varphi^{-1} (A)) \\
&=C(u;L')\lambda_{M}(A).
\end{aligned}
\]
The last equality arises from the change of variables formula, as $\lambda_{M}=\lambda_{\mathcal{C}}$ is the pushforward of $\lambda^{2}$ onto $\mathcal{C}$ under $\varphi$, and hence it follows that the cumulant function of $L$ pushed forward onto $M$ coincides with that of the Lévy basis $L_{M}$ defined on $(M,\mathcal{M})$ with CQ $(\gamma,\Sigma,\nu,\lambda_{M})$. It thus holds that the pushforward of $L$ onto $M$ under $\varphi$ is equal in law to $L_{M}$.
\end{proof}
\paragraph{Proof of Proposition \ref{prop:levy_ito}}
\begin{proof}
Let $L$ be a Lévy basis on $(\mathbb{R}^2,\mathcal{B}(\mathbb{R}^2))$ with CQ $(\gamma,\Sigma,\nu,\lambda^2)$. By Theorem 4.5 in \textcite{Pedersen2003}, there exists a modification $L^{*}$ of $L$ which admits the following Lévy-Itô decomposition for $B\in\mathcal{B}_{b}(\mathbb{R}^2)$
\begin{equation}\label{eq:levy_ito2}
L^{*}(B) = \gamma \lambda^{2}(B) + W (B) + \int_{\mathbb{R}}y\mathbf{1}_{(-1,1)}(y)(N-\ell)(d y,B) + \int_{\mathbb{R}}y\mathbf{1}_{(-1,1)^\complement}(y)N(d y,B),
\end{equation}
where $W$ is a Lévy basis with CQ $(0,\Sigma,0,\lambda^2)$ and $N$ a Lévy basis with CQ $(0,0,\nu,\lambda^2)$ and $\ell$ its compensator. Applying Proposition \ref{prop:pushforward} to the right hand side of \eqref{eq:levy_ito2}, it follows that we can define $L_{M}=L^{*}\circ \varphi^{-1}$, which then admits the Lévy-Itô decomposition \eqref{eq:levy_ito1} and has CQ $(\gamma,\Sigma,\nu,\lambda_M)$.
\end{proof}
\paragraph{Proof of Proposition \ref{prop:ambit_stationary}}
\begin{proof}
Take $g\in G$. As $g$ is an isometry, the Riemannian metric on $\mathcal{C}$ is preserved under $g$, and it thus follows from \eqref{eq:volume_integral} that $\lambda_{\mathcal{C}}(g(A)) = \lambda_{\mathcal{C}}(A)$ for any $A\in \mathcal{M}$. It follows from Proposition~\ref{prop:levy_basis_stationary} that $L$ is stationary with respect to $G$ cf. Definition~\ref{def:stationary}. To show stationarity of the ambit field $Y_{t}(h)$ in the stated sense, note that 
\[
\begin{aligned}
Y_{t+\tau}(R_{\theta}h) &= \int_{A_{t+\tau}(R_{\theta}h)}\kappa (t+\tau,s,R_{\theta}h,\xi) a_{s}(\xi)\lambda_{\mathcal{C}}(ds,d\xi) \\&\quad + \int_{A_{t+\tau}(R_{\theta}h)}K (t+\tau,s,R_{\theta}h,\xi) \sigma_{s}(\xi)L(ds,d\xi) \\
&= \int_{A_{t}(h)}\kappa (t-\zeta,\theta_h-\theta_\eta) a_{\tau+\zeta}(R_{\theta}\eta)\lambda_{\mathcal{C}}(d(\zeta+\tau),d(R_{\theta}\eta))\\&\quad + \int_{A_{t}(h)}K (t-\zeta,\theta_h-\theta_\eta) \sigma_{\tau+\zeta}(R_{\theta}\eta)L(d(\zeta+\tau),d(R_{\theta}\eta)) \\
&\overset{d}{=} \int_{A_{t}(h)}\kappa (t-\zeta,\theta_h-\theta_\eta) a_{\zeta}(\eta)\lambda_{\mathcal{C}}(d\zeta,d\eta) \\&\quad + \int_{A_{t}(h)}K (t-\zeta,\theta_h-\theta_\eta) \sigma_{\zeta}(\eta)L(d\zeta,d\eta) \\
&= Y_{t}(h),
\end{aligned}
\]
where we have used stationarity of $L$ with respect to $G$ combined with invariance of $\lambda_\mathcal{C}$ under isometries in the third equality. It follows that the imposed assumptions in the Proposition are sufficient for stationarity of the field $Y_{t}(h)$. To assess necessity, we note that the points $h,\xi$ are fully determined by their angular coordinates $\theta_h,\theta_\xi$. The kernels $\kappa,K$ must depend only on rotationally invariant quantities, which is thus equivalent to depending only on the quantity $\theta_h-\theta_\xi$.
\end{proof}

\paragraph{Proof of Proposition \ref{prop:ambit_field_cumulant}}
The first equality in \eqref{eq:cumulant_function} follows from the general result of \textcite{Ambit}, Proposition~40 or \textcite{RecentAdvances}, Proposition~1. The second equality follows from \eqref{eq:volume_integral}.

\paragraph{Proof of Proposition \ref{prop:second_order_structure}}
\begin{proof}
Proposition 41 in \textcite{Ambit} yields that
\[
\begin{aligned}
\mathbb{E}\left[ D_{t}(h) \right] &= \mathbb{E}\left[ L' \right]\int_{A_{t}(h)}K(t,s,h,\xi)\mathbb{E}\left[ \sigma_{s}(\xi)\right] \lambda_{\mathcal{C}}(ds,d\xi), \\[1.25ex]
\mathrm{Cov}\left( D_{t}(h), D_{t'}(h') \right) &= \mathbb{V}\left[ L'\right]\int_{A_{t}(h)\cap A_{t'}(h')}K(t,s,h,\xi)K(t',s,h',\xi')\mathbb{E}\left[ \sigma_{s}(r(\theta ))^{2} \right]\lambda_{\mathcal{C}}(ds,d\xi) \\
+\mathbb{E}\left[ (L')^2 \right]\int_{A_{t}(h)}&\int_{A_{t'}(h')}K(t,s,h,\xi)K(t',s',h',\xi')\varrho (s,s',\xi, \xi')\lambda_{\mathcal{C}}(ds,d\xi)\lambda_\mathcal{C}(ds',d\xi').
\end{aligned}
\]
Now note that $A_{t}(h)\cap A_{t'}(h') = A_{\min \lbrace t, t'\rbrace}(h)$ and apply \eqref{eq:volume_integral} to obtain the stated result.
\end{proof}

\paragraph{Proof of Proposition \ref{prop:esscher_transform}}
\begin{proof}
Note that $\mathbb{E}^{\mathbb{P}}[Z_t]=1$, and that it is an $\mathcal{F}_t$ martingale, since
\[
\begin{aligned}
\mathbb{E}^{\mathbb{P}}\left[ Z_t \mid \mathcal{F}_\tau\right] &= \mathbb{E}^{\mathbb{P}}\left[ \exp\left( \int_{\tau}^{t}\int_{\Pi}q(s,\xi)L(ds,d\xi)\right)\right]\exp\left(-\int_{\tau}^{t}\int_{\Pi}C(q(s,\xi);L')\lambda_{\mathcal{C}}(ds,d\xi) \right) \\
&\quad \times \exp\left( \int_{0}^{\tau}\int_{\Pi}q(s,\xi)L(ds,d\xi) - \int_{0}^{\tau}\int_{\Pi}C(q(s,\xi);L') \lambda_{\mathcal{C}}(ds,d\xi) \right) \\
&= Z_{\tau}.
\end{aligned}
\]
It follows that $Z_t$ defines a valid density process. By the independently scattered property of $L$ under $\mathbb{P}$, we have for $A\in\mathcal{B}_{b}(\Pi)$ that
\[
\begin{aligned}
\mathbb{E}^{\mathbb{P}^{q}}\left[ e^{uL([0,t]\times A)} \right] &= \mathbb{E}^{\mathbb{P}}\left[ \exp\left( \int_{0}^{t}\int_{A}u+q(s,\xi)L(ds,d\xi) - \int_{0}^{t}\int_{A}C(q(s,\xi);L')\lambda_{\mathcal{C}}(ds,d\xi) \right) \right] \\
&\quad \times \mathbb{E}^{\mathbb{P}}\left[ \int_{0}^{t}\int_{\Pi\setminus A}q(s,\xi)L(ds,d\xi) \right] \exp\left( \int_{0}^{t}\int_{\Pi\setminus A}C(q(s,\xi);L')\lambda_{\mathcal{C}}(ds,d\xi) \right) \\
&= \exp\left( \int_{0}^{t}\int_{A}C(u+q(s,\xi);L')\lambda_{\mathcal{C}}(ds,d\xi) - \int_{0}^{t}\int_{A}C(q(s,\xi);L')\lambda_{\mathcal{C}}(ds,d\xi) \right).
\end{aligned}
\] 
It follows that the distribution of $L([0,t]\times A)$ under $\mathbb{P}^q$ is infinitely divisible with triplet $(\tilde{\gamma},\Sigma,\tilde{\nu})$ as in \eqref{eq:esscher_triplet}. The intensity measure of $L$ is unaffected and remains $\lambda_{\mathcal{C}}$. By similar reasoning, we may show that for $u_{1},u_{2},\ldots ,u_{N}$ and disjoint $A_{1},A_{2},\ldots ,A_{N}$ in $\mathcal{B}_{b}([0,T^\ast]\times \Pi)$ such that each $\mathbb{E}[e^{u_{j}L(A_j)}]<\infty$, it holds that
\[
\mathbb{E}^{\mathbb{P}^{q}}\left[ \exp\left( \sum_{j=1}^{N}u_{j}L(A_j) \right) \right] = \prod_{j=1}^{N}\mathbb{E}^{\mathbb{P}^q}\left[ e^{u_{j}L(A_j)} \right].
\]
Hence $L$ is independently scattered under $\mathbb{P}^{q}$ and thus a Lévy basis under $\mathbb{P}^q$ with CQ $(\tilde{\gamma},\Sigma,\tilde{\nu},\mathcal{\lambda}_{\mathcal{C}})$.
\end{proof}

\paragraph{Proof of Proposition \ref{prop:D_spectral_rep}}
\begin{proof}
Expanding the kernel $K$ as in \eqref{eq:repr_K_L2} and assuming that we may freely interchange sums and integrals, it holds that 
\[
\begin{aligned}
X_{n}(t) &= \int_{0}^{2\pi}D_{t}(r(\theta))e^{-\mathrm{i}n\theta}d\theta \\
&\overset{L^2}{=} \frac{1}{2\pi}\sum_{j_1,j_2,j_3}\int_{\Pi}e^{-\mathrm{i}n r(\theta)}\int_{-\infty}^{t}\int_{\Pi}c_{j_1,j_2,j_3}\Psi_{j}(t-s)e^{\mathrm{i}j_2 h}e^{\mathrm{i}j_3\xi}\sigma_{s}(\xi) L(ds,d\xi) d\theta \\
&= \frac{1}{2\pi}\sum_{j_1,j_2,j_3}\left( \int_{\Pi}e^{\mathrm{i}j_2 h}e^{-\mathrm{i}n h} \lambda_{\Pi}(dh) \right)c_{j_1,j_2,j_3}\int_{-\infty}^{t}\int_{\Pi}\Psi_{j_1}(t-s)e^{\mathrm{i}j_3 \xi}\sigma_{s}(\xi)L(ds,d\xi) \\
&= \sum_{j_1,j_3}c_{j_1,n,j_3}\int_{-\infty}^{t}\int_{\Pi}\Psi_{j_1}(t-s)e^{\mathrm{i}j_3 \xi}\sigma_{s}(\xi) L(ds,d\xi).
\end{aligned}
\]
It follows immediately from the expansion \eqref{eq:repr_K_L2} that $\sum_{n\in\mathbb{Z}}X_{n}(t)e^{\mathrm{i}nh}\overset{L^2}{=}2\pi D_{t}(h)$. We now argue that we may in fact swap the infinite sum and integrals. Consider two finite-order approximations $K^{(\mathbf{N_1})},K^{(\mathbf{N_2})}$ of $K$. The Itô isometry of \textcite{Walsh1986} yields that (dropping the integration arguments for brevity)
\begin{equation}\label{eq:Xn_representation}
\mathbb{E}\left[ \left( \int K^{(\mathbf{N_1})} - K^{(\mathbf{N_2})} \right)^2 \right] = \int \left( K^{(\mathbf{N_1})} - K^{(\mathbf{N_2})} \right)^2 \underset{\lvert N_1\rvert,\lvert N_2\rvert\to 0}{\longrightarrow} 0,
\end{equation}
where the convergence follows from completeness of $L^{2}$. By linearity we may swap the sum and integrals for the finite order approximations $K^{(\mathbf{N_1})},K^{(\mathbf{N_2})}$ and letting $\lvert \mathbf{N_1}\vert,\lvert \mathbf{N_2}\rvert\to \infty$ completes the argument. 
\end{proof}

\paragraph{Proof of Proposition \ref{prop:spectral_density}}
\begin{proof}
By the assumption of isotropy, it holds that $c_{j_1,j_2,j_3}=0$ in the expansion \eqref{eq:repr_K_L2} whenever $j_3\neq -j_2$. The representation \eqref{eq:repr_K_L2} thus becomes
\[
K(t-s,\theta-\phi) = \frac{1}{2\pi}\sum_{j_1,\ell}c_{j_1,-\ell,\ell}\Psi_{j_1}(t-s) e^{-\mathrm{i}\ell(\theta- \phi)}.
\]
Since $\mathbb{E}[L']=0$, it follows from Proposition~\ref{prop:second_order_structure} that
\[
\begin{aligned}
\mathrm{Cov}(X_{n}(t+\tau),\overline{X_{n}(t)}) &= \mathbb{E}[X_{n}(t+\tau)\overline{X_{n}(t)}] \\
&= \mathbb{V}[L']\mathbb{E}[\sigma^2]\int_{-\infty}^{t}\int_{\Pi}\left(\sum_{j}c_{j,-n,n}\Psi_{j}(t+\tau-s)\right)\overline{\left( \sum_{j}c_{j,-n,n}\Psi_{j}(t-s) \right)}\lambda_{\mathcal{C}}(ds,d\xi).
\end{aligned}
\]
Let $a_{j,n}=c_{j,-n-,n}$ and define the transfer function $H_{n}(v)$ as 
\[
H_{n}(v) = \int_{0}^{\infty}X_{n}(\tau)e^{-\mathrm{i}v\tau }d\tau = \sum_{j}a_{j,n}\widehat{\Psi}_{j}(v),
\]
where $\widehat{\Psi}_j(v)= \int_{0}^\infty \Psi_j(t)e^{-\mathrm{i}vt}dt$ and the last equality follows from \eqref{eq:Xn_representation}. Then it follows immediately by the Fubini theorem that
\[
\begin{aligned}
f_{n}(u) &= \frac{1}{2\pi}\int_{\mathbb{R}}e^{-\mathrm{i}u\tau}\mathbb{E}[X_{n}(t+\tau)\overline{X_n(t)}]d\tau \\ &= 
\frac{\mathbb{V}[L']\mathbb{E}[\sigma^2 ]}{2\pi}\int_{0}^{\infty}\overline{\sum_{j}c_{j,-n,n}\Psi_j(r)}\int_{\mathbb{R}}\sum_{j}c_{j,-n,n}\Psi_{j}(r+\tau)d\tau dr \\
&= \frac{\mathbb{V}[L']\mathbb{E}[\sigma^2 ]}{2\pi}\int_{0}^{\infty}\overline{\sum_{j}c_{j,-n,n}\Psi_j(r)} e^{\mathrm{i}ur}H_{n}(u) dr \\
&= \frac{\mathbb{V}[L']\mathbb{E}[\sigma^2 ]}{2\pi}\overline{H_{n}(u)}H_{n}(u) \\
&= \frac{\mathbb{V}[L']\mathbb{E}[\sigma^2 ]}{2\pi}\lvert H_{n}(u)\rvert^2.
\end{aligned}
\]
It remains to show that $\widehat{\Psi}_j(u)$ is given by \eqref{eq:Psi_fourier_transform}. Using that the generalized Laguerre polynomials have a representation in terms of the confluent hypergeometric function as $\mathcal{L}_{j}^{(\alpha)}(t)=\frac{(\alpha+1)_j}{j!}{}_{1}F_{1}(-j;\alpha+1;t)$ where $(x)_j$ denotes the Pochhammer symbol and combining this with the standard Laplace transform identity
\[
\int_{0}^\infty t^\gamma e^{-\zeta t}{}_{1}F_{1}(-j;\alpha+1;t)dt = \Gamma(\gamma+1)\zeta^{-(\gamma +1)}	{}_{2}F_{1}(-j,\alpha+1;\alpha+1;\tfrac{1}{\zeta}), 
\]
gives \eqref{eq:Psi_fourier_transform} upon setting $\gamma=\frac{\alpha}{2}$ and $\zeta = (\frac{1}{2}+\mathrm{i}v)$.
\end{proof}

\paragraph{Proof of Proposition \ref{prop:identifiability}}
\begin{proof}
We first want to show that if two admissible kernels $K_1,K_2$ yield the same spectrum $(f_{n}(u))_{n\in\mathbb{Z}}$, then $K_1=K_2$ in $L^2$. By isotropy, the decomposition \eqref{eq:Xn_representation} of Proposition~\ref{prop:D_spectral_rep} diagonalizes since each $c_{n_1,n_2,n_3}=0$ whenever $n_2\neq -n_3$. As a consequence, the $X_{n}(t)$ processes are orthogonal as processes in the sense that $\mathrm{Cov}(X_{n}(t),\overline{X}_{m}(t)) = 0$ for $n\neq m$. The second order structure of each $X_{n}(t)$ uniquely determines the spectral density $f_{n}$, and hence the second order structure of the full field $D_{t}(h)$ uniquely determines $(f_{n})_{n\in\mathbb{Z}}$. This also holds for the kernel functions $K_{1},K_{2}$, since their orthonormal basis representations are unique. I.e., a kernel $K$ uniquely determines a spectrum $(f_{n})_{n\in \mathbb{Z}}$ via the associated system of transfer functions $(H_{n})_{n\in\mathbb{Z}}$. However, the reverse map $(f_{n})_{n\in\mathbb{Z}}\to K$ is not injective, since $f_{n}$ is only identified up to the modulus $\lvert H_{n}\rvert$ and a scaling constant. The assumption that $c_{0,0}=1$ normalizes the kernel, such that the spectral density is in fact identified up to $\lvert H_{n}(u)\rvert$, with a fixed scaling constant. Overall, the spectral density thus determines the kernel $K$ up to the equivalence class of transfer functions with identical modulus. Now suppose that $K_1$ and $K_2$ determines the systems of transfer functions $(H_{n}^{(1)})_{n\in\mathbb{Z}}$ and $(H_{n}^{(2)})_{n\in\mathbb{Z}}$, respectively, and that $\lvert H_{n}^{(1)}(v)\rvert = \lvert H_{n}^{(2)}(v)\rvert$ for all $n$ and $v$. Then we can define
\[
A_{n}(v) = \frac{H_{n}^{(2)}(v)}{H_{n}^{(1)}(v)},
\]
which satisfies $\lvert A_{n}(v) \rvert = 1$. By the minimum-phase assumption, we can extend each $H_{n}^{(i)}(v)$ analytically to the half-plane $\Im(v) < 0$ where they have no zeros and thus $A_{n}(v)$ admits an analytical extension to that half-plane. Thus $A_n$ is an inner function in the Hardy space $\mathcal{H}^2$. The minimum-phase assumption also entails that the transfer functions $H_{n}^{(i)}$ are outer functions in $\mathcal{H}^2$. Since every function in $\mathcal{H}^2$ can be uniquely represented as a product of an inner and an outer function, the only way that the ratio $A_{n}$ can be inner, is if it is trivially inner, i.e., $A_{n}(v)=e^{\mathrm{i}x_{n}}$ for some $x_n\in\mathbb{R}$ and all $v$ in the half-plane where $\Im(v)<0$. It must therefore hold that
\[
H_{n}^{(1)}(v)=e^{\mathrm{i}x_n}H_{n}^{(2)}.
\]
By taking the inverse Fourier transforms, it follows that the coefficients $a_{j,n}^{(1)}$ and $a_{j,n}^{(2)}$ are related via $
a_{j,n}^{(1)}=e^{\mathrm{i}x_n}a_{j,n}^{(2)}$.
This holds in particular for the coefficients 
$
a_{j^{\ast}(n),n}^{(1)} = e^{\mathrm{i}x_n}a_{j^{\ast}(n),n}^{(2)}
$ and since the $a_{j^\ast (n),n}^{(i)}$'s are real-valued and strictly positive, we must have that $e^{\mathrm{i}x_n}=1$, in which case it holds that $H_{n}^{(1)}(v) = H_{n}^{(2)}(v)$ for all $n$ and $v$ and that $a_{j,n}^{(1)}=a_{j,n}^{(2)}$ for all $j,n$. Since the kernels $K_{1},K_{2}$ are uniquely identified by the coefficients $a_{j,n}$, we find that $K_{1}=K_{2}$ in $L^2$, which was to be shown. The same statement for the truncated kernels $K^{J,N}$ follows by truncating the representation \eqref{eq:repr_K_L2}.
\end{proof}

\paragraph{Proof of Proposition \ref{prop:simulation}}
\begin{proof}
The function $\xi\mapsto K(t-s,h,\xi)$ has Fourier coefficients 
\[
\widetilde{K}(t-s,h,n) = \frac{1}{2\pi}\int_{0}^{2\pi}e^{\mathrm{i}\phi n}K(t-s,h,r(\phi))d\phi, \quad n \in \mathbb{Z}.
\]
By the Fourier inversion theorem, it follows that 
\[
K(t-s,h,\xi) = \sum_{n\in \mathbb{Z}}\widetilde{K}(t-s,h,n)e^{\mathrm{i}\xi n}.
\]
Due to Assumption \ref{ass:laplace_kernel}, we may represent $\tilde{K}$ by Laplace inversion in $(t-s)$. Utilizing the Bromwich integral representation of the Laplace inversion theorem, we have
\[
\tilde{K}(t-s,h,n) = \frac{1}{2\mathrm{i}\pi}\int_{\mathbb{R}}\widehat{K}(z_r+\mathrm{i}z_i,h,n) e^{(z_r + \mathrm{i}z_{i})(t-s)}dz_i, 
\]
with $\widehat{K}(z,h,n)$ defined as in \eqref{eq:kernel_transforms}, and $z_{r}\in (-\infty,\gamma)$. By the stochastic Fubini theorem we have
\[
\begin{aligned}
Y_{t}(h) &= \int_{0}^{t}\int_{\Pi}\sum_{n\in \mathbb{Z}}\int_{\mathbb{R}}\widehat{K}(z_r+\mathrm{i}z_i,h,n) e^{\mathrm{i}\xi n}e^{(z_r+\mathrm{i}z_i)(t-s)}dz L(ds,d\xi) \\[1.25ex]
&= \sum_{n\in \mathbb{Z}}\int_{\mathbb{R}}\widehat{K}(z_r + \mathrm{i}z_i, h, n)V_{n}(t,z) dz,
\end{aligned}
\]
with $V_n (t,z)$ as in \eqref{eq:Vn_OU}. That $t\mapsto V_{n}(t,z)$ defines a complex-valued Ornstein-Uhlenbeck type process follows from the Fubini theorem of Proposition 45 in \textcite{Ambit}, which we may use to obtain
\[
\begin{aligned}
z\int_{0}^{t}V_{n}(s,z)ds &= z\int_{0}^{t}\int_{0}^{s}\int_{\Pi}e^{z(s-v)}e^{\mathrm{i}n\xi}\sigma_{v}(\xi)L(dv,d\xi)ds \\
&= z\int_{0}^{t}\int_{\Pi}\int_{0}^{t}\mathbf{1}_{[0,s]}(v)e^{z(s-v)}ds e^{\mathrm{i}n\xi}\sigma_{v}(\xi) L(dv,d\xi) \\
&= \int_{0}^{t}\int_{\Pi}\left( e^{z(t-v)}-1 \right)e^{\mathrm{i}n\xi}\sigma_{v}(\xi)L(dv,d\xi) \\
&= V_{n}(t,z) - L_{n}(t),
\end{aligned}
\]
where $L_{n}(t)$ is the noise driving the Ornstein-Uhlenbeck process given by
\[
L_{n}(t) = \int_{0}^{t}\int_{\Pi}e^{\mathrm{i}n\xi}\sigma_{s}(\xi)L(ds,d\xi).
\]
\end{proof}
\paragraph{Proof of Proposition \ref{prop:error_bound}}
\begin{proof}
Define the kernel $K^{(N)}$ as
\[
K^{(N)}(t,h,n) = \frac{1}{2\mathrm{i}\pi}\sum_{n\in \lbrace -N,\ldots ,N \rbrace} \int_{\mathbb{R}}\hat{K}(z_{r}+\mathrm{i}z_{i}h,n)e^{(z_r+\mathrm{i}z_i)t}e^{\mathrm{i}n\xi}dz_{i}.
\]
Define the covariance measure $Q$ of $L$ as
\[
Q([0,t]\times A) = \left(\Sigma + \int_{\mathbb{R}}x^2\nu (d x)\right)\lambda_{\mathcal{C}}([0,t]\times A).
\]
Applying the Itô isometry of \textcite{Walsh1986}, we then obtain
\[
\begin{aligned}
\mathbb{E}\left[ \left( Y_{t}(h)-Y_{t}^{(N)}(h)\right)^2 \right] &= \int_{0}^{t}\int_{\Pi}\left( K(t-s,h,\xi) - K^{(N)}(t-s,h,\xi) \right)^2 \mathbb{E}\left[ \sigma_{s}(\xi)^2\right] Q(ds,d\xi)  \\
&\leq k_{\sigma}\int_{0}^{t}\int_{\Pi}\left( K(t-s,h,\xi) - K^{(N)}(t-s,h,\xi) \right)^2 Q(ds,d\xi).
\end{aligned}
\]
Note then that
\[
K(t-s,h,\xi) - K^{(N)}(t-s,h,\xi) = \sum_{n\in \mathbb{Z}\setminus \lbrace -N,\ldots ,N\rbrace}\int_{\mathbb{R}}\frac{1}{2\mathrm{i}\pi}\hat{K}(z_{r}+\mathrm{i}z_{i}h,n)e^{(z_r+\mathrm{i}z_i)(t-s)}e^{\mathrm{i}n\xi}dz_{i},
\]
where we have moved the factor $\frac{1}{2\mathrm{i}\pi}$ into the integral to signify that it is part of the integral operator, and that the output $K(t-s,h,\xi)$ is real-valued. Using Tonelli's theorem, we have
\[
\begin{aligned}
\int_{0}^{t}\int_{\Pi}&\left( K(t-s,h,\xi) - K^{(N)}(t-s,h,\xi) \right)^2 Q(ds,d\xi) \\
&\leq\frac{1}{4\pi^2} \int_{0}^{t}\int_{\Pi}\left(\sum_{n\in \mathbb{Z}\setminus \lbrace -N,\ldots ,N\rbrace}\left\lvert\int_{\mathbb{R}}\hat{K}(z_{r}+\mathrm{i}z_{i}h,n)e^{(z_r+\mathrm{i}z_i)(t-s)}e^{\mathrm{i}n\xi}dz_{i}\right\rvert\right)^2 Q(ds,d\xi) \\
&= \frac{\mathbb{V}[L']}{4\pi^2}\sum_{n\in \mathbb{Z}\setminus \lbrace -N,\ldots,N\rbrace}\int_{0}^{t}\left(\left\lvert\int_{\mathbb{R}}\hat{K}(z_{r}+\mathrm{i}z_{i}h,n)e^{(z_r+\mathrm{i}z_i)(t-s)}dz_{i}\right\rvert\right)^2 ds \\
&\leq \frac{\mathbb{V}[L']}{4\pi^2}\sum_{n\in \mathbb{Z}\setminus \lbrace -N,\ldots,N\rbrace}\left(\int_{\mathbb{R}}\lvert \hat{K}(z_{r}+\mathrm{i}z_{i},h,n)\rvert dz_{i}\right)^2 \int_{0}^{t}e^{2z_{r}s}ds.
\end{aligned}
\]
If we select the integration contour of $\hat{K}$ such that $z_{r}<0$ (which we may do, so long as $\gamma < 0$), then we can bound the integral of $e^{2z_{r}s}$ by $\frac{1}{2\lvert z_{r}\rvert}$, and obtain \eqref{eq:simulation_bound}.
\end{proof}

\section{De-seasonalization of German spot price data}\label{app:data}
We use the notation $P_t(d)$ as in Section~\ref{sec:model_setup} for the observed price on time $t$ and delivery period $d$. We suppose that $P_{t}(d)$ decomposes additively as
\[
P_{t}(d) = \mu_{t}(d) + \zeta_{t}(d) + \varepsilon_{t}(d),
\] 
where $\mu_{t}(d)$ represents a smoothly varying trend, $\zeta_{t}(d)$ a periodic component representing seasonality, and $\varepsilon_t(d)$ the unexplained residuals, representing unanticipated price shocks. The trend $\mu_t(d)$ is extracted as a 90-day backward-looking moving average. This leaves 24 time series $P_{t}(d)-\widehat{\mu}_{t}(d)$. The seasonal component $\zeta_{t}(d)$ is extracted via the multi-period STL (MSTL) algorithm of \textcite{MSTL}, which is an extension of the well-known seasonal-trend decomposition algorithm of \textcite{STL}. The algorithm requires period lengths as an input, where we select periods of length $1/4$ and $1$, corresponding to quarterly and yearly seasonality. Intuitively, we expect to extract predictable patterns that vary with the four seasons and over the course of a year in a periodic fashion. It is common to extract such seasonal components by regressing the de-trended price series onto a sum of trigonometric functions as in \textcites{Meyer-BrandisTankov2008}{Bennedsen2017}, but the MSTL algorithm permits variation in the amplitudes of the periodic component across periods, which can be thought of as allowing for ``stochastic seasonality''. We depict the extracted surfaces of trend and seasonal components on Figure~\ref{fig:season_trend_decomp}, where it is evident that the seasonal component varies substantially over time. The residual series $\varepsilon_{t}(d)$ are all tested for unit roots using the augmented Dickey-Fuller test with automatic lag selection and for stationarity via the KPSS test. In all cases, the results strongly favour the assumption of stationary. Furthermore, for each residual time series we perform a simple intercept-only linear regression with heteroskedastic and autocorrelation adjusted standard errors; all tests fail to reject the null hypothesis that the stationary mean is zero. We therefore treat the residual series as realizations of the zero-mean and time stationary field $D_{t}(h)$.
\begin{figure}[htbp]
	\centering
	\begin{subfigure}[b]{\textwidth}
		\centering
		\includegraphics[width=0.9\textwidth]{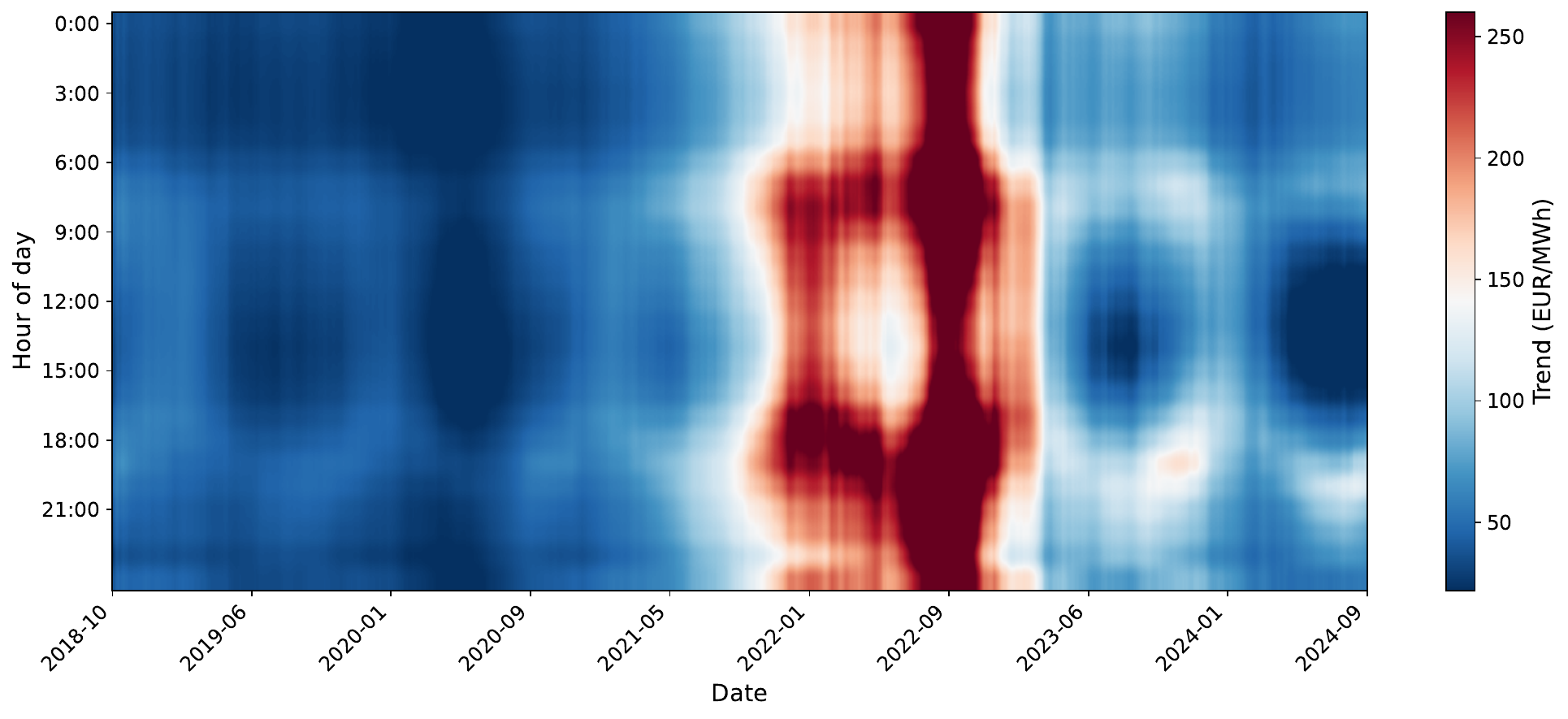}
	\end{subfigure}
	
	\vspace{0.3cm}
	
	\begin{subfigure}[b]{\textwidth}
		\centering
		\includegraphics[width=0.9\textwidth]{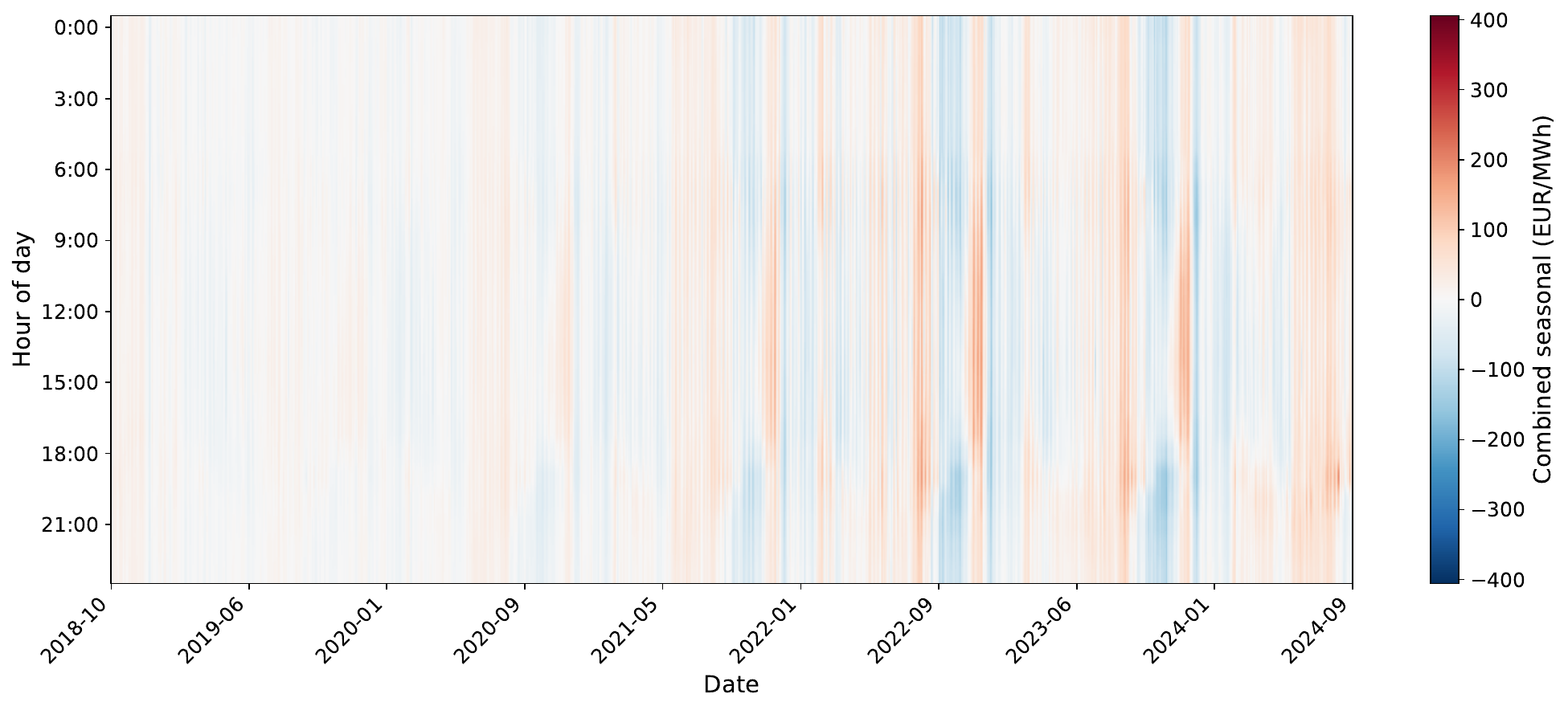}
	\end{subfigure}
	
	\vspace{0.3cm}
	\begin{subfigure}[b]{\textwidth}
		\centering
		\includegraphics[width=0.9\textwidth]{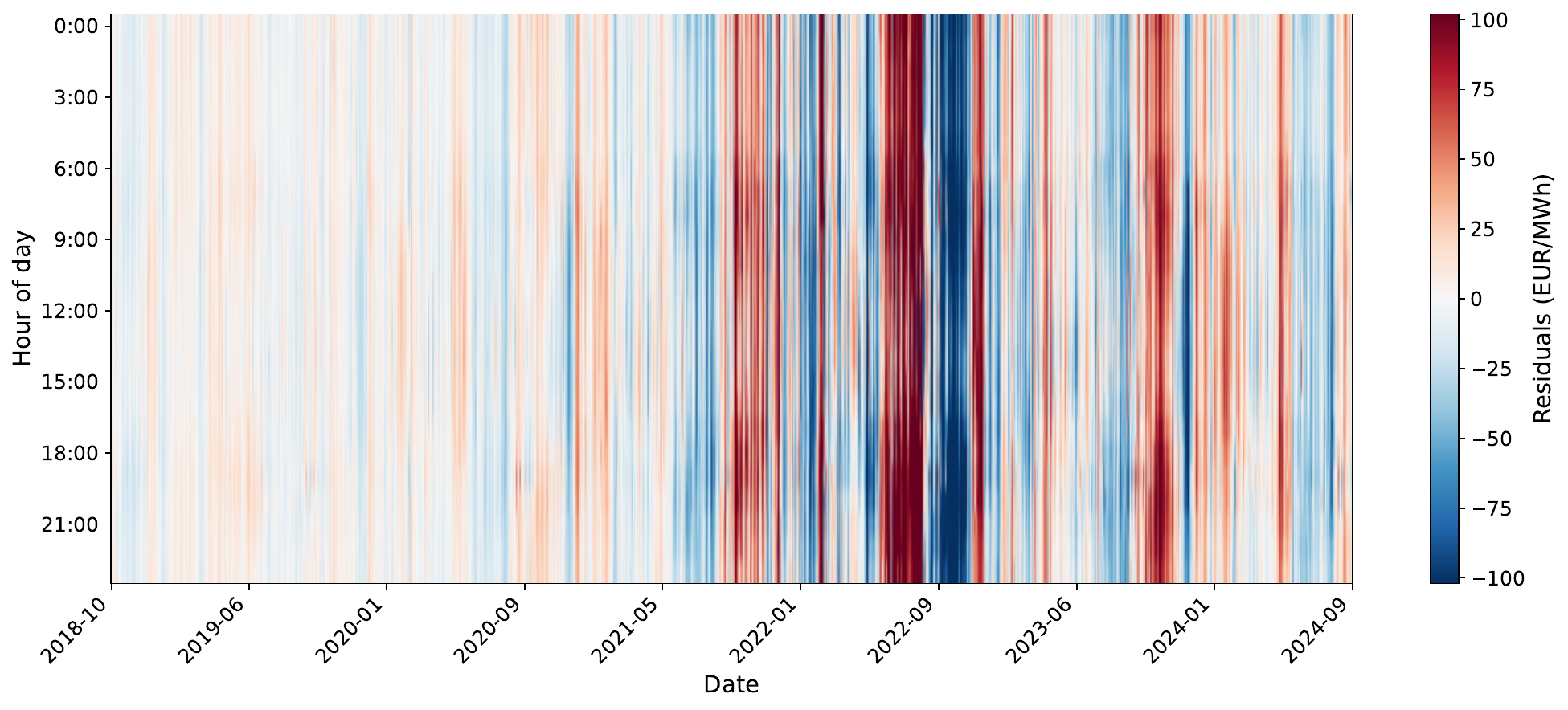}
	\end{subfigure}
	\caption{Trend (top), seasonal (middle), and residual (bottom) components from the MSTL decomposition of German day-ahead electricity prices.}\label{fig:season_trend_decomp}
\end{figure}

\pagebreak
\section{Spectral density fits}\label{app:plot}
\begin{figure}[!h]
	\begin{center}
	\begin{subfigure}[b]{0.48\textwidth}
		\includegraphics[scale=0.375]{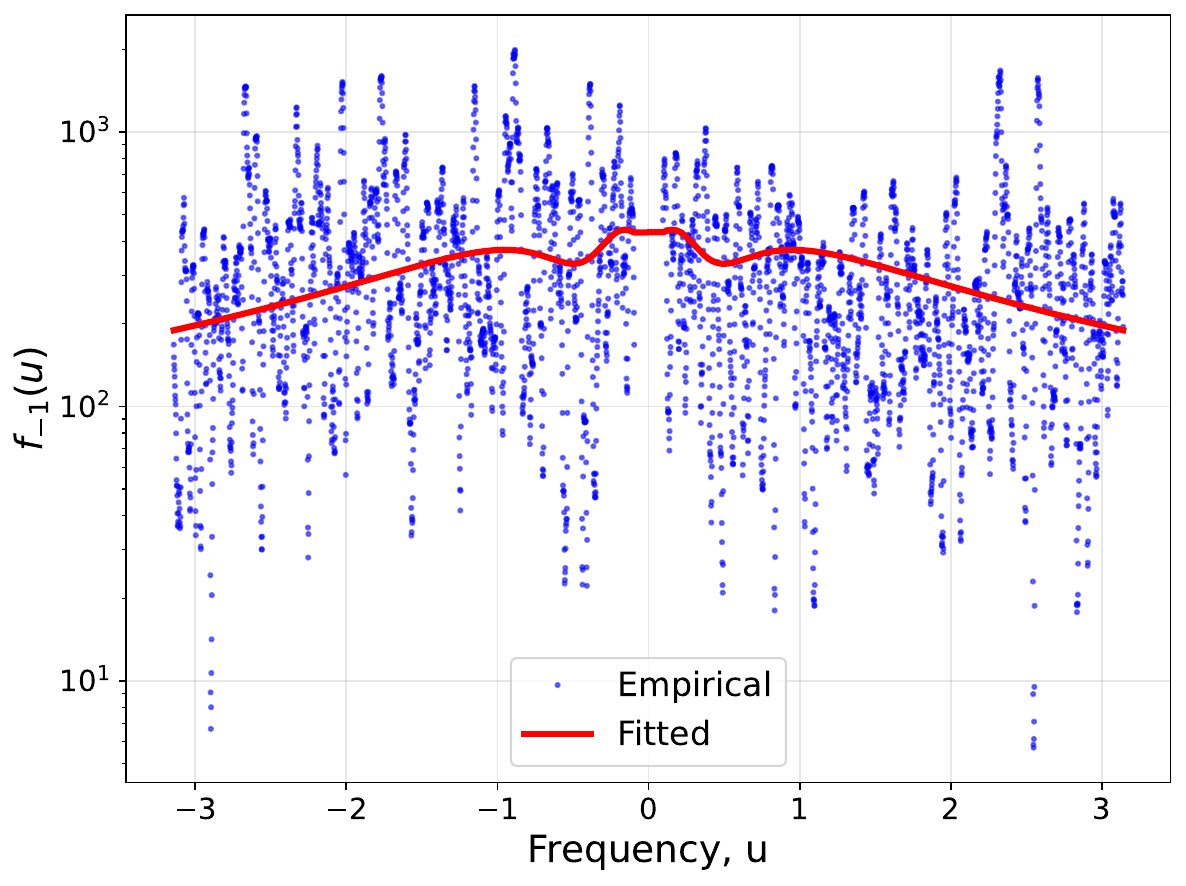}
	\end{subfigure}
	\begin{subfigure}[b]{0.48\textwidth}
		\includegraphics[scale=0.375]{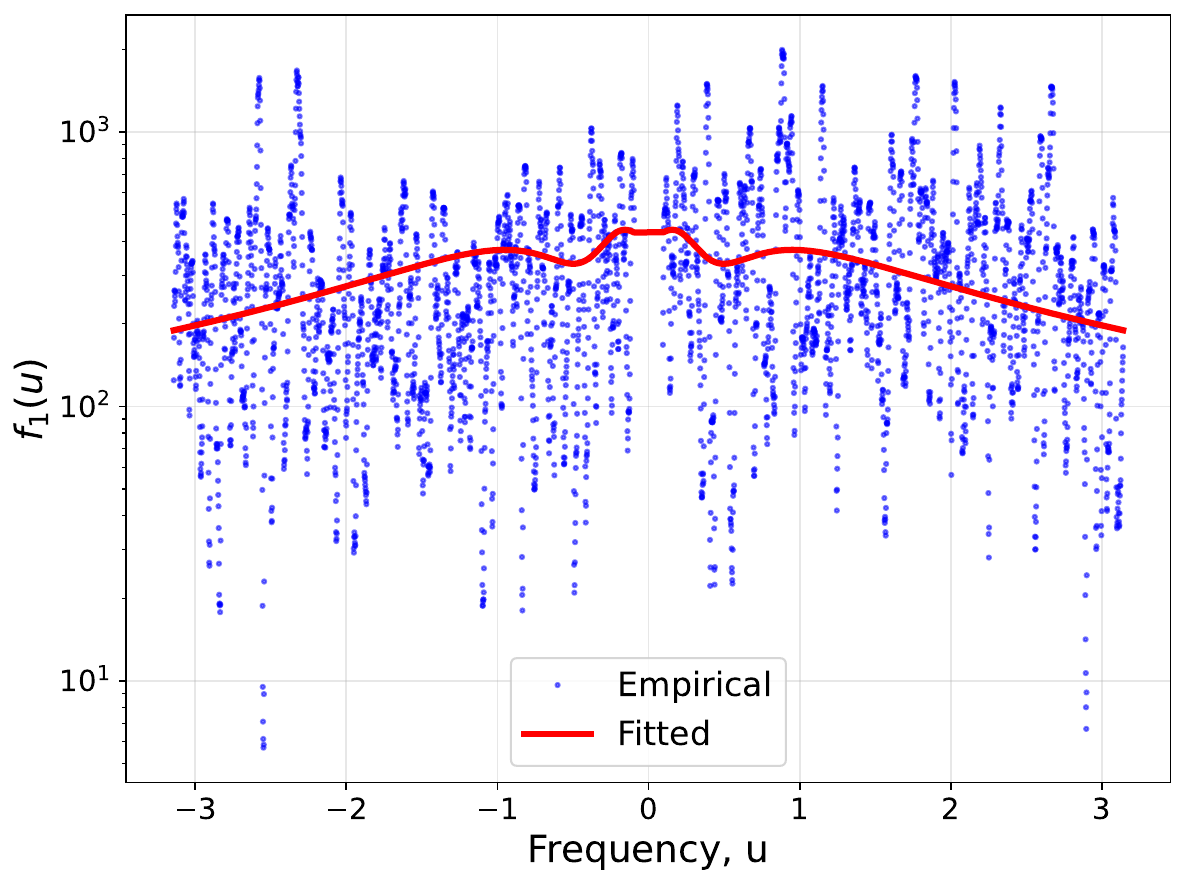}
	\end{subfigure}
	
	\vspace{0.4cm}
	
	\begin{subfigure}[b]{0.48\textwidth}
		\includegraphics[scale=0.375]{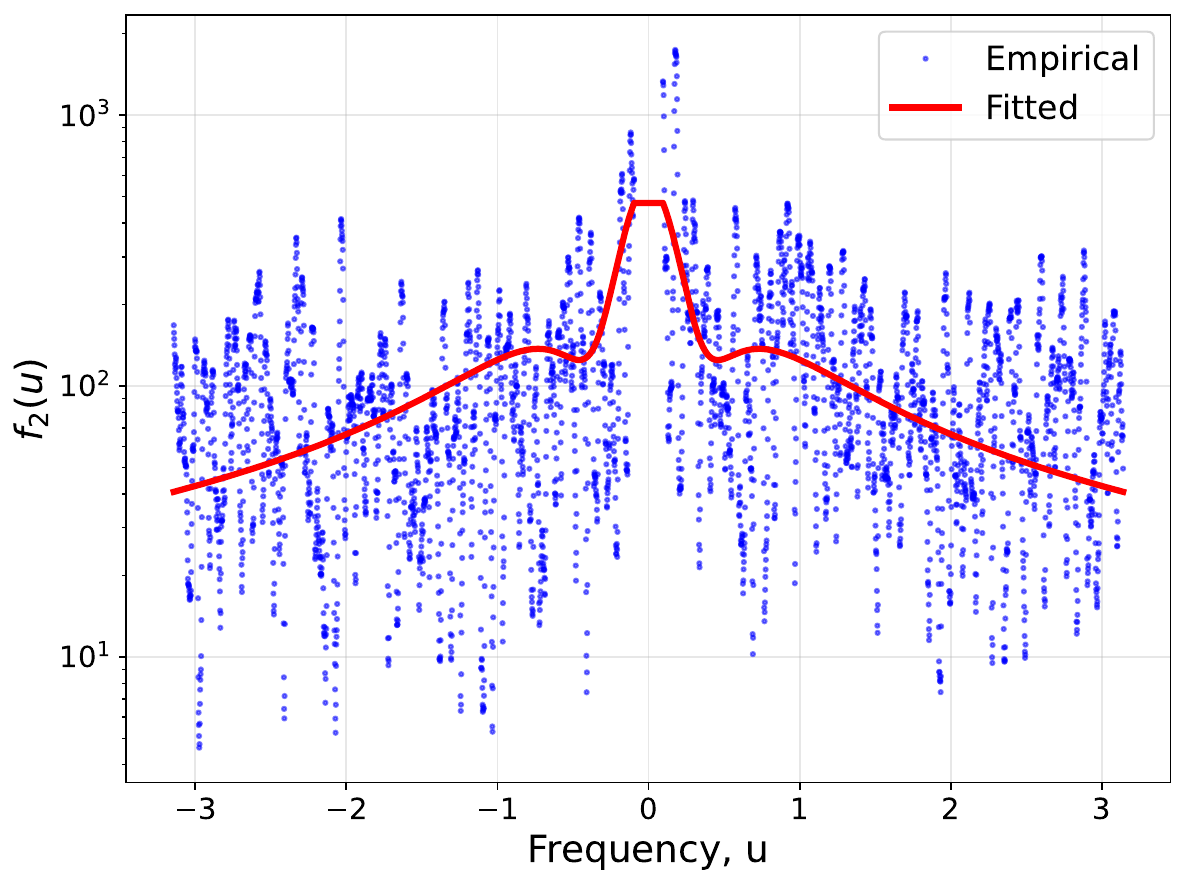}
	\end{subfigure}
	\begin{subfigure}[b]{0.48\textwidth}
		\includegraphics[scale=0.375]{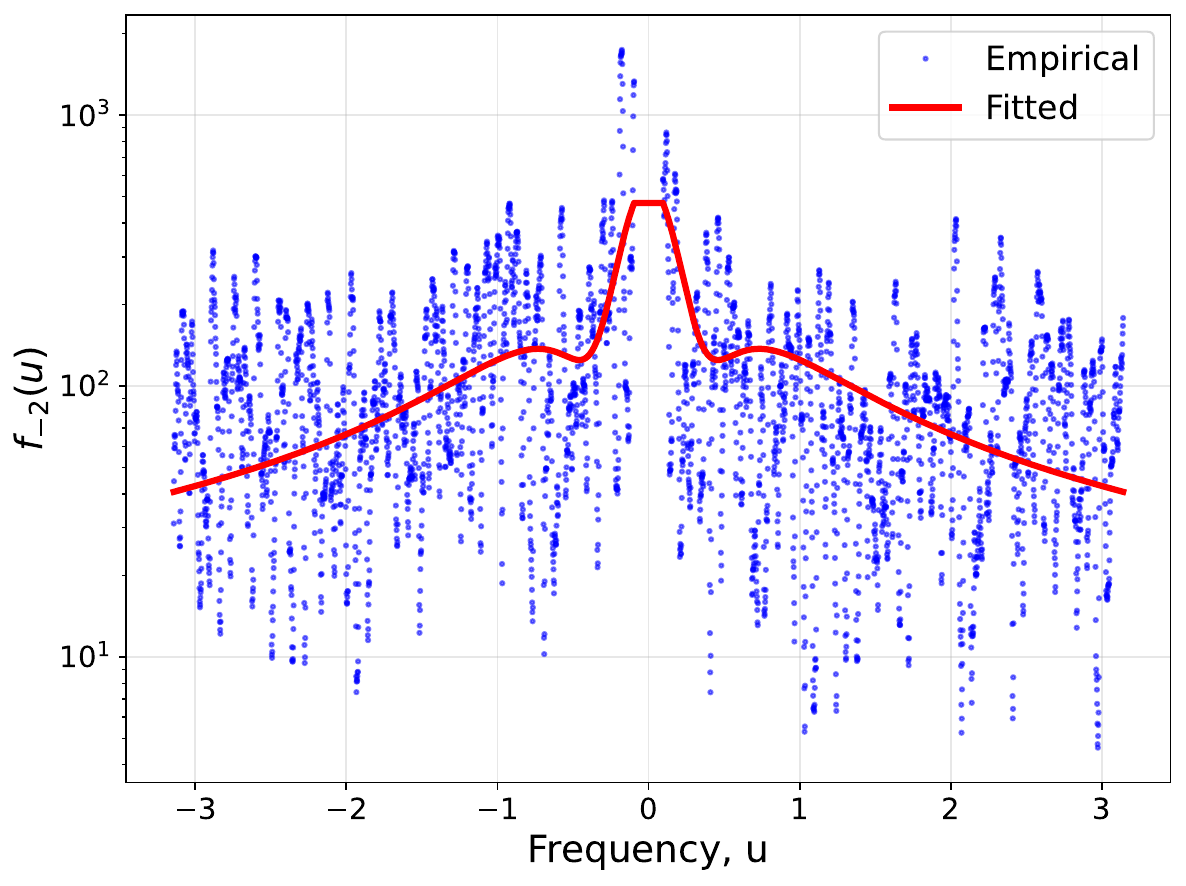}
	\end{subfigure}
	\end{center}
	\caption{Fitted spectral density at modes $n=1$ (upper left), $n=-1$ (upper right), $n=2$ (lower left), and $n=-2$ (lower right) against the empirical counterpart (right).}\label{fig:kernel_fit_appendix}
\end{figure}

\section{Simulation algorithm and implementation details}\label{app:simulation}
\begin{example}[Simulation algorithm]\label{ex:simulation_algorithm}
Let the setting be as in Proposition \ref{prop:simulation}, and fix an equispaced time-discretization $0=t_{0}<t_{1}<\cdots <t_{J}=T$, a truncation order $n\in \lbrace -N,\ldots,N\rbrace$, and a space-discretization $0<h_{1}<\cdots < h_{H}=2\pi$. Let also $Z$ be a grid of numbers in $\mathbb{R}$, representing a discretization of the integration range. 
\begin{enumerate}
\item For each $n$, simulate the complex-valued OU process $V_{n}(t,z)$ iteratively as follows.
\begin{itemize}
\item Initialize $V_{n}(0,z) = 0$ for all values of $z\in Z$ and $L_{n}(t_{0})=0$.
\item Compute the contribution of $L_{n}(t)$ as
\[
(L_{n}(t_{j})-L_{n}(t_{j-1})) \approx \sum_{l=1}^{H}e^{\mathrm{i}nh_{l}}\lambda_{\mathcal{C}}\left([t_{j}-t_{j-1}]\times [h_{l}-h_{l-1}]\right)L'_{l},
\]
where $L_{l}'$ are i.i.d. copies of the Lévy seed of $L$.
\item Obtain the next step $V_{n}(t_{j},z)$ via \eqref{eq:sim_complex_OU}, for $j=1,\ldots, J$ and $z\in Z$.
\end{itemize}
\item For each $n$, compute the Fourier coefficients $\hat{K}(z,h,n)$ for all $z\in Z$.
\item For all $j$ and $l$, numerically approximate $I_{j,l}=\int_{\mathbb{R}}\hat{K}(z_{r}+\mathrm{i}z_{i},h_{l},n)V_{n}(t_{j},z_{r}+\mathrm{i}z_{i})dz_{i}$ over the grid $Z$, for example by a trapezoidal rule.
\item For all $j$ and $l$, compute the value of the ambit field $Y_{t_{j}}(h_{l})$ as 
\[
Y_{t_{k}}(h_{l}) = \sum_{n\in \lbrace -N,\ldots ,N\rbrace}I_{j,l},
\]
where $Y_{0}(h_{l})=0$ for each $l=1,\ldots ,H$.
\end{enumerate}
\end{example}

\end{document}